%% file: ms.tex
\documentclass[twocolumn, floatfix, nofootinbib, aps, pra, 10pt]{revtex4-2}

\input{packages}

\input{definitions}

\begin{document}
\title{No Invariant Perfect Qubit Codes}

\author{Refik Mansuroglu \orcidlink{0000-0001-7352-513X}}
\email[]{Refik.Mansuroglu@fau.de}
\affiliation{Department of Physics, Friedrich-Alexander-Universität Erlangen-Nürnberg (FAU), Staudtstraße 7, 91058 Erlangen}
\author{Hanno Sahlmann \orcidlink{0000-0002-8083-7139}}
\email[]{Hanno.Sahlmann@gravity.fau.de}
\affiliation{Department of Physics, Friedrich-Alexander-Universität Erlangen-Nürnberg (FAU), Staudtstraße 7, 91058 Erlangen}

\date{\today}

\begin{abstract}
    Perfect tensors describe highly entangled quantum states that have attracted particular attention in the fields of quantum information theory and quantum gravity. In loop quantum gravity, the natural question arises whether SU(2) invariant tensors, which are fundamental ingredients of the basis states of spacetime, can also be perfect. In this work, we present a number of general constraints for the layout of such invariant perfect tensors (IPTs) and further describe a systematic and constructive approach to check the existence of an IPT of given valence. We apply our algorithm to show that no qubit encoding of valence \rema{6} can be described by an IPT and close a gap to prove a no-go theorem for invariant perfect qubit encodings. \rema{We also provide two alternative proofs for the non-existence of 4-valent qubit IPTs which has been shown in \cite{HIGUCHI2000213,Li_2017}.}
\end{abstract}

\maketitle

\section{Introduction}
\input{Chapters/Introduction}

\section{Basic notions and Concepts}
\input{Chapters/Basic_Def}

\section{Necessary Conditions on Invariant Perfect Tensors}
\label{sec:limits}
\input{Chapters/Limitations_on_spin}

\section{Master Equation for Invariant Perfect Tensors}
\label{sec:main_master}
\input{Chapters/sufficient_cond}

\section{Two No-Go Theorems for valence Four and Six}
\label{sec:nogo}
\input{Chapters/nogo}

\section{Conclusion}
\input{Chapters/Conclusion}

\begin{acknowledgments}

\rema{We thank Otfried Gühne and Fabian Bernards for discussions, as well as the anonymous referees for their feedback to improve the presentation of this paper.} This work is part of the Munich Quantum Valley, which is supported by the Bavarian state government with funds from the Hightech Agenda Bayern Plus.
\end{acknowledgments}

\twocolumngrid
\bibliography{literature}

\newpage

\onecolumngrid
\appendix
\section{No 4-valent invariant perfect tensors -- Proof I}
\input{Chapters/no_4_v_irred}

\section{Derivation of Master Equations for Invariant Perfect Tensors}
\input{Chapters/thetaNet}
\input{Chapters/master_eq}

\section{Additional Equations from Repartition}
\input{Chapters/Repartitions}

\section{No 4-valent invariant perfect tensors -- Proof II}
\input{Chapters/no_4v}

\section{No 6-valent invariant perfect tensors}
\input{Chapters/no_6v}

\end{document}

%% file: packages.tex
\usepackage[utf8]{inputenc}
\usepackage[english]{babel}
\usepackage{natbib}
\usepackage{graphicx}
\usepackage{amsmath, amssymb, amsthm}
\usepackage{tensor}
\usepackage{braket}
\usepackage{cases}
\usepackage{dsfont}
\usepackage[dvipsnames]{xcolor}
\usepackage{hyperref}
\hypersetup{
    colorlinks = true,
    allcolors = RoyalPurple,
    breaklinks=true
}
\usepackage[percent]{overpic}
\usepackage{svg}
\usepackage{tikz}
\usepackage[export]{adjustbox}
\usepackage{leftidx}
\usepackage{BOONDOX-cal} 
\usepackage{orcidlink}

%% file: definitions.tex
\numberwithin{equation}{section}

\DeclareMathOperator{\Tr}{Tr}

\DeclareMathOperator{\Real}{Re}

\newtheorem{defi}{Definition}[section]
\newtheorem{rem}[defi]{Remark}
\newtheorem{prop}[defi]{Lemma}
\newtheorem{theo}[defi]{Theorem}
\newtheorem{coro}[defi]{Corollary}
\newtheorem{exa}[defi]{Example}
\newtheorem{alg}[defi]{Algorithm}

\newcommand{\rema}[1]{#1}
\newcommand{\change}[1]{#1}

\newcommand{\vstate}[2]{
\adjustbox{valign=b}{
    \begin{tikzpicture}
        \draw[black, thick] (0,0) -- (3,0);
        \draw[black, thick] (0,.1) -- (1.25,.1);
        \draw[black, thick] (0,.5) -- (1,.5);
        \draw[black, thick] (1.75,.1) -- (3,.1);
        \draw[black, thick] (2,.5) -- (3,.5);
        \draw[black, thick] (1,.5) -- (1.3,0);
        \draw[black, thick] (1.7,0) -- (2,.5);
        \draw[black] (0.5,0.4) circle (0pt) node[anchor=center]{$\ddots$};
        \draw[black] (2.5,0.4) circle (0pt) node[anchor=center]{\reflectbox{$\ddots$}};
        \draw[black] (1.3,0.4) circle (0pt) node[anchor=center]{#1};
        \draw[black] (1.7,0.4) circle (0pt) node[anchor=center]{#2};
    \end{tikzpicture}
    }
}

\newcommand{\thetaNet}[2]{
\adjustbox{valign=b}{
    \begin{tikzpicture}
        \draw[black, thick] (0,0) ellipse (.2cm and .3cm);
        \draw[black, thick] (-0.1,.2) -- (0.1,.2);
        \draw[black, thick] (-0.1,-.2) -- (0.1,-.2);
        \draw[black] (0,0.1) circle (0pt) node[anchor=center]{$\vdots$};
        \draw[black] (-.3,0) circle (0pt) node[anchor=center]{#1};
        \draw[black] (0.3,0) circle (0pt) node[anchor=center]{#2};
    \end{tikzpicture}
    }
}

\newcommand{\loopstate}[4]{
\adjustbox{valign=c}{
    \begin{tikzpicture}
        \draw[black, thick] (-1,0) -- (0,0);
        \draw[black, thick] (0.5,0) ellipse (.5cm and .2cm);
        \draw[black, thick] (1.0,0) -- (2,0);
        \draw[black] (-.5,0.3) circle (0pt) node[anchor=center]{#1};
        \draw[black] (1.5,0.3) circle (0pt) node[anchor=center]{#2};
        \draw[black] (0.5,0.5) circle (0pt) node[anchor=center]{#3};
        \draw[black] (0.5,-0.5) circle (0pt) node[anchor=center]{#4};
    \end{tikzpicture}
    }
}

\newcommand{\identity}[1]{
\adjustbox{valign=b}{
    \begin{tikzpicture}
        \draw[black, thick] (-1,0) -- (1,0);
        \draw[black] (0,0.3) circle (0pt) node[anchor=center]{#1};
    \end{tikzpicture}
    }
}

\newcommand{\identities}[0]{
\adjustbox{valign=c}{
    \begin{tikzpicture}
        \draw[black, thick] (-1,-0.1) -- (0,-0.1);
        \draw[black, thick] (-1,0.1) -- (0,0.1);
        \draw[black] (-0.5,0.4) circle (0pt) node[anchor=center]{$\vdots$};
        \draw[black, thick] (-1,0.5) -- (0,0.5);
        \draw[black, thick] (-1,0.7) -- (0,0.7);
    \end{tikzpicture}
    }
}

\newcommand{\binorMin}[4]{
\adjustbox{valign=c}{
    \begin{tikzpicture}
        \draw[black, thick] (-1,0) -- (-0.5,0);
        \draw[black, thick] (-0.5,-.5) rectangle ++(0.2,1);
        \draw[black, thick] (-0.3,0) -- (0.5,0);
        \draw[black, thick] (-0.3,-0.2) -- (0,-0.5);
        \draw[black, thick] (0,-.7) rectangle ++(0.2,0.5);
        \draw[black, thick] (0.2,-0.5) -- (0.5,-0.2);
        \draw[black, thick] (0.5,-0.5) rectangle ++(0.2,1);
        \draw[black, thick] (0.7,0) -- (1.2,0);
        \draw[black] (-0.75,0.3) circle (0pt) node[anchor=center]{#1};
        \draw[black] (0.95,0.3) circle (0pt) node[anchor=center]{#2};
        \draw[black] (0.1,0.3) circle (0pt) node[anchor=center]{#3};
        \draw[black] (0.1,-1) circle (0pt) node[anchor=center]{#4};
    \end{tikzpicture}
    }
}

\newcommand{\binorMax}[4]{
\adjustbox{valign=c}{
    \begin{tikzpicture}
        \draw[black, thick] (-1,-0.5) -- (-0.5,-0.5);
        \draw[black, thick] (-0.5,-.75) rectangle ++(0.2,0.5);
        \draw (0.2,0) arc
        [
            start angle=-70,
            end angle=250,
            x radius=.3cm,
            y radius =.2cm
        ] ;
        \draw[black, thick] (-0.3,-0.5) -- (0,-0.5);
        \draw[black, thick] (0,-.75) rectangle ++(0.2,1);
        \draw[black, thick] (0.2,-0.5) -- (0.5,-0.5);
        \draw[black, thick] (0.5,-0.75) rectangle ++(0.2,0.5);
        \draw[black, thick] (0.7,-0.5) -- (1.2,-0.5);
        \draw[black] (-0.75,0) circle (0pt) node[anchor=center]{#1};
        \draw[black] (0.95,0) circle (0pt) node[anchor=center]{#2};
        \draw[black] (0.1,0.7) circle (0pt) node[anchor=center]{#3};
        \draw[black] (0.1,-1) circle (0pt) node[anchor=center]{#4};
    \end{tikzpicture}
    }
}

\newcommand{\thetaLoop}[4]{
\adjustbox{valign=b}{
    \begin{tikzpicture}
        \draw[black, thick] (-1,0) -- (2,0);
        \draw (1,0) arc
        [
            start angle=0,
            end angle=180,
            x radius=.5cm,
            y radius =1cm
        ] ;
        \draw[black, thick] (0.17,0.75) -- (0.83,0.75);
        \draw[black, thick] (0.02,0.25) -- (0.98,0.25);
        \draw[black] (0.5,0.5) circle (0pt) node[anchor=center]{...};
        \draw[black] (-.5,0.3) circle (0pt) node[anchor=center]{#1};
        \draw[black] (1.5,0.3) circle (0pt) node[anchor=center]{#2};
        \draw[black] (-0.1,0.5) circle (0pt) node[anchor=center]{#3};
        \draw[black] (1.1,0.5) circle (0pt) node[anchor=center]{#4};
    \end{tikzpicture}
    }
}

\newcommand{\vnorm}[4]{
\adjustbox{valign=b}{
    \begin{tikzpicture}
        \draw[densely dotted, black, thick] (0,0) -- (0.3,0);
        \draw[densely dotted, black, thick] (0,.1) -- (0.3,.1);
        \draw[densely dotted, black, thick] (0,.5) -- (0.3,.5);
        \draw[black, thick] (0.3,0) -- (4.5,0);
        \draw[black, thick] (0.3,.1) -- (1.25,.1);
        \draw[black, thick] (0.3,.5) -- (1,.5);
        \draw[black, thick] (1.75,.1) -- (2.95,.1);
        \draw[black, thick] (2,.5) -- (2.7,.5);
        \draw[black, thick] (1,.5) -- (1.3,0);
        \draw[black, thick] (1.7,0) -- (2,.5);
        \draw[black, thick] (2.7,0.5) -- (3,0);
        \draw[black, thick] (3.5,0) -- (3.8,0.5);
        \draw[black, thick] (3.55,.1) -- (4.5,.1);
        \draw[black, thick] (3.8,.5) -- (4.5,.5);
        \draw[densely dotted, black, thick] (4.5,0) -- (4.8,0);
        \draw[densely dotted, black, thick] (4.5,.1) -- (4.8,.1);
        \draw[densely dotted, black, thick] (4.5,.5) -- (4.8,.5);
        \draw[black] (0.8,0.4) circle (0pt) node[anchor=center]{$\ddots$};
        \draw[black] (2.35,0.4) circle (0pt) node[anchor=center]{$\vdots$};
        \draw[black] (4,0.4) circle (0pt) node[anchor=center]{\reflectbox{$\ddots$}};
        \draw[black] (1.3,0.4) circle (0pt) node[anchor=center]{#1};
        \draw[black] (1.7,0.4) circle (0pt) node[anchor=center]{#2};
        \draw[black] (3.05,0.4) circle (0pt) node[anchor=center]{#3};
        \draw[black] (3.55,0.4) circle (0pt) node[anchor=center]{#4};
    \end{tikzpicture}
    }
}

\newcommand{\vdagv}[4]{
\adjustbox{valign=b}{
    \begin{tikzpicture}
        \draw[black, thick] (0.3,0) -- (4.5,0);
        \draw[black, thick] (0.3,.1) -- (1.25,.1);
        \draw[black, thick] (0.3,.5) -- (1,.5);
        \draw[black, thick] (1.75,.1) -- (2.95,.1);
        \draw[black, thick] (2,.5) -- (2.7,.5);
        \draw[black, thick] (1,.5) -- (1.3,0);
        \draw[black, thick] (1.7,0) -- (2,.5);
        \draw[black, thick] (2.7,0.5) -- (3,0);
        \draw[black, thick] (3.5,0) -- (3.8,0.5);
        \draw[black, thick] (3.55,.1) -- (4.5,.1);
        \draw[black, thick] (3.8,.5) -- (4.5,.5);
        \draw[black] (0.8,0.4) circle (0pt) node[anchor=center]{$\ddots$};
        \draw[black] (2.35,0.4) circle (0pt) node[anchor=center]{$\vdots$};
        \draw[black] (4,0.4) circle (0pt) node[anchor=center]{\reflectbox{$\ddots$}};
        \draw[black] (1.3,0.4) circle (0pt) node[anchor=center]{#1};
        \draw[black] (1.7,0.4) circle (0pt) node[anchor=center]{#2};
        \draw[black] (3.05,0.4) circle (0pt) node[anchor=center]{#3};
        \draw[black] (3.55,0.4) circle (0pt) node[anchor=center]{#4};
    \end{tikzpicture}
    }
}

\newcommand{\invtensor}[2]{
\adjustbox{valign=b}{
    \begin{tikzpicture}
        \draw[black, thick] (0,0) ellipse (.3cm and .3cm);
        \draw[black, thick] (-0.5,.3) -- (-0.25,.2);
        \draw[black, thick] (-0.25,-.2) -- (-0.5,-.3);
        \draw[black] (-0.5,0.1) circle (0pt) node[anchor=center]{$\vdots$};
        \draw[black] (0.5,0.1) circle (0pt) node[anchor=center]{$\vdots$};
        \draw[black] (1.5,0.1) circle (0pt) node[anchor=center]{$\vdots$};
        \draw[black, thick] (1,0) ellipse (.3cm and .3cm);
        \draw[black, thick] (1.5,.3) -- (1.25,.2);
        \draw[black, thick] (1.25,-.2) -- (1.5,-.3);
        \draw[black, thick] (0.25,.2) -- (0.75,.2);
        \draw[black, thick] (0.25,-.2) -- (0.75,-.2);
        \draw[black] (0,0) circle (0pt) node[anchor=center]{#1};
        \draw[black] (1,0) circle (0pt) node[anchor=center]{#2};
    \end{tikzpicture}
    }
}

\newcommand{\trivialone}[2]{
\adjustbox{valign=b}{
    \begin{tikzpicture}
        \draw[black, thick] (0,0) ellipse (.3cm and .3cm);
        \draw[black, thick] (-0.5,0) -- (-0.25,.2);
        \draw[black, thick] (-0.5,0.15) -- (-0.3,.05);
        \draw[black, thick] (-0.25,-.2) -- (-0.5,-.3);
        \draw[black] (0.5,0.1) circle (0pt) node[anchor=center]{$\vdots$};
        \draw[black, thick] (1,0) ellipse (.3cm and .3cm);
        \draw[black, thick] (1.5,0.2) -- (1.3,0);
        \draw[black, thick] (1.5,0.05) -- (1.25,0.2);
        \draw[black, thick] (1.25,-.2) -- (1.5,-.3);
        \draw[black, thick] (0.25,.2) -- (0.75,.2);
        \draw[black, thick] (0.25,-.2) -- (0.75,-.2);
        \draw[black] (0,0) circle (0pt) node[anchor=center]{#1};
        \draw[black] (1,0) circle (0pt) node[anchor=center]{#2};
    \end{tikzpicture}
    }
}

\newcommand{\identitiesone}[0]{
\adjustbox{valign=c}{
    \begin{tikzpicture}
        \draw[black, thick] (-1,-0.1) -- (0,-0.1);
        \draw[black, thick] (-1,0.1) -- (0,0.1);
        \draw[black] (-0.5,0.4) circle (0pt) node[anchor=center]{$\vdots$};
        \draw[black, thick] (-0.8,0.5) -- (-0.2,0.5);
        \draw[black, thick] (-0.8,0.7) -- (-0.2,0.7);
        \draw[black, thick] (-1,0.7) -- (-0.8,0.5);
        \draw[black, thick] (-1,0.5) -- (-0.8,0.7);
        \draw[black, thick] (-0.2,0.7) -- (0,0.5);
        \draw[black, thick] (-0.2,0.5) -- (0,0.7);
    \end{tikzpicture}
    }
}

\newcommand{\trivialtwo}[2]{
\adjustbox{valign=b}{
    \begin{tikzpicture}
        \draw[black, thick] (0,0) ellipse (.3cm and .3cm);
        \draw[black, thick] (-0.5,0.3) -- (-0.25,.2);
        \draw[black, thick] (-0.25,-.2) -- (-0.5,-.3);
        \draw[black] (-0.5,0.1) circle (0pt) node[anchor=center]{$\vdots$};
        \draw[black] (2.5,0.1) circle (0pt) node[anchor=center]{$\vdots$};
        \draw[black, thick] (2,0) ellipse (.3cm and .3cm);
        \draw[black, thick] (2.5,0.3) -- (2.25,0.2);
        \draw[black, thick] (2.25,-.2) -- (2.5,-.3);
        \draw[black, thick] (0.3,0) -- (0.6,.2);
        \draw[black, thick] (0.25,.2) -- (0.6,0);
        \draw[black, thick] (0.6,.2) -- (1.4,.2);
        \draw[black, thick] (0.6,0) -- (1.4,0);
        \draw[black, thick] (1.4,0) -- (1.75,.2);
        \draw[black, thick] (1.4,.2) -- (1.7,0);
        \draw[black, thick] (0.25,-.2) -- (1.75,-.2);
        \draw[black] (0,0) circle (0pt) node[anchor=center]{#1};
        \draw[black] (2,0) circle (0pt) node[anchor=center]{#2};
    \end{tikzpicture}
    }
}

\newcommand{\bipartstate}[1]{
\adjustbox{valign=c}{
    \begin{tikzpicture}
        \draw[black, thick] (-0.5,0.25) -- (-0.2,0.25);
        \draw[black, thick] (-0.5,0) -- (-0.3,0);
        \draw[black, thick] (0,0) ellipse (.3cm and .3cm);
        \draw[black, thick] (0.5,0.25) -- (0.2,0.25);
        \draw[black, thick] (0.5,0) -- (0.3,0);
        \draw[black] (0,0) circle (0pt) node[anchor=center]{#1};
        \draw[black, thick] (-.2,-0.25) arc
        [
            start angle=150,
            end angle=290,
            x radius=.5cm,
            y radius =.1cm
        ] ;
        \draw[black, thick] (.2,-0.25) arc
        [
            start angle=30,
            end angle=-110,
            x radius=.5cm,
            y radius =.1cm
        ] ;
    \end{tikzpicture}
    }
}

\newcommand{\bipartproduct}[2]{
\adjustbox{valign=c}{
    \begin{tikzpicture}
        \draw[black, thick] (-0.5,0.25) -- (-0.2,0.25);
        \draw[black, thick] (-0.5,0) -- (-0.3,0);
        \draw[black, thick] (0,0) ellipse (.3cm and .3cm);
        \draw[black, thick] (0.65,0.25) -- (0.2,0.25);
        \draw[black, thick] (0.5,0) -- (0.3,0);
        \draw[black] (0,0) circle (0pt) node[anchor=center]{#1};
        \draw[black, thick] (-.2,-0.25) arc
        [
            start angle=150,
            end angle=400,
            x radius=.25cm,
            y radius =.1cm
        ] ;
        \draw[black, thick] (0.8,0) ellipse (.3cm and .3cm);
        \draw[black, thick] (1,0.25) -- (1.4,0.25);
        \draw[black, thick] (1.1,0) -- (1.4,0);
        \draw[black] (0.8,0) circle (0pt) node[anchor=center]{#2};
        \draw[black, thick] (.6,-0.25) arc
        [
            start angle=150,
            end angle=400,
            x radius=.25cm,
            y radius =.1cm
        ] ;
        \draw[black, thick] (-0.5,-0.5) -- (1.4,-0.5);
    \end{tikzpicture}
    }
}

\newcommand{\bipartproductII}[2]{
\adjustbox{valign=c}{
    \begin{tikzpicture}
        \draw[black, thick] (-0.5,0.25) -- (-0.2,0.25);
        \draw[black, thick] (-0.5,0) -- (-0.3,0);
        \draw[black, thick] (0,0) ellipse (.3cm and .3cm);
        \draw[black, thick] (0.65,0.25) -- (0.2,0.25);
        \draw[black, thick] (0.5,0) -- (0.3,0);
        \draw[black] (0,0) circle (0pt) node[anchor=center]{#1};
        \draw[black, thick] (-0.5,-0.5) -- (-0.2,-0.25);
        \draw[black, thick] (0.2,-0.25) -- (0.5,-0.5);
        \draw[black, thick] (0.5,-0.5) -- (1.4,-0.5);
        \draw[black, thick] (0.8,0) ellipse (.3cm and .3cm);
        \draw[black, thick] (1,0.25) -- (1.4,0.25);
        \draw[black, thick] (1.1,0) -- (1.4,0);
        \draw[black] (0.8,0) circle (0pt) node[anchor=center]{#2};
        \draw[black, thick] (.6,-0.25) arc
        [
            start angle=150,
            end angle=400,
            x radius=.25cm,
            y radius =.1cm
        ] ;
        
    \end{tikzpicture}
    }
}

\newcommand{\bipartproductIII}[2]{
\adjustbox{valign=c}{
    \begin{tikzpicture}
        \draw[black, thick] (-0.5,0.25) -- (-0.2,0.25);
        \draw[black, thick] (-0.5,0) -- (-0.3,0);
        \draw[black, thick] (0,0) ellipse (.3cm and .3cm);
        \draw[black, thick] (0.65,0.25) -- (0.2,0.25);
        \draw[black, thick] (0.5,0) -- (0.3,0);
        \draw[black] (0,0) circle (0pt) node[anchor=center]{#1};
        \draw[black, thick] (0.5,-0.5) -- (0.65,-0.25);
        \draw[black, thick] (1,-0.25) -- (1.4,-0.5);
        \draw[black, thick] (-0.5,-0.5) -- (0.5,-0.5);
        \draw[black, thick] (0.8,0) ellipse (.3cm and .3cm);
        \draw[black, thick] (1,0.25) -- (1.4,0.25);
        \draw[black, thick] (1.1,0) -- (1.4,0);
        \draw[black] (0.8,0) circle (0pt) node[anchor=center]{#2};
        \draw[black, thick] (-.2,-0.25) arc
        [
            start angle=150,
            end angle=400,
            x radius=.25cm,
            y radius =.1cm
        ] ;
        
    \end{tikzpicture}
    }
}

\newcommand{\bipartproductIV}[2]{
\adjustbox{valign=c}{
    \begin{tikzpicture}
        \draw[black, thick] (-0.5,0.25) -- (-0.2,0.25);
        \draw[black, thick] (-0.5,0) -- (-0.3,0);
        \draw[black, thick] (0,0) ellipse (.3cm and .3cm);
        \draw[black, thick] (0.6,0.25) -- (0.2,0.25);
        \draw[black, thick] (0.5,0) -- (0.3,0);
        \draw[black] (0,0) circle (0pt) node[anchor=center]{#1};
        \draw[black, thick] (-.2,-0.25) arc
        [
            start angle=150,
            end angle=290,
            x radius=.5cm,
            y radius =.1cm
        ] ;
        \draw[black, thick] (.2,-0.25) arc
        [
            start angle=30,
            end angle=-110,
            x radius=.5cm,
            y radius =.1cm
        ] ;
        \draw[black, thick] (0.8,0) ellipse (.3cm and .3cm);
        \draw[black] (0.8,0) circle (0pt) node[anchor=center]{#2};
        \draw[black, thick] (.6,-0.25) arc
        [
            start angle=150,
            end angle=290,
            x radius=.5cm,
            y radius =.1cm
        ] ;
        \draw[black, thick] (1,-0.25) arc
        [
            start angle=30,
            end angle=-110,
            x radius=.5cm,
            y radius =.1cm
        ] ;
        \draw[black, thick] (1,0.25) -- (1.4,0.25);
        \draw[black, thick] (1.1,0) -- (1.4,0);
    \end{tikzpicture}
    }
}

\newcommand{\fourvrepart}[1]{
\adjustbox{valign=c}{
    \begin{tikzpicture}
        \draw[black, thick] (0,.5) -- (0.5,.5);
        \draw[black, thick] (2,.5) -- (2.5,.5);
        \draw[black, thick] (1,0) -- (1.5,0);
        \draw[black, thick] (0.5,.5) -- (1,0);
        \draw[black, thick] (1.5,0) -- (2,.5);
        \draw[black] (1.3,0.2) circle (0pt) node[anchor=center]{#1};
        \draw[black, thick] (1,0) arc
        [
            start angle=150,
            end angle=290,
            x radius=1cm,
            y radius =.1cm
        ] ;
        \draw[black, thick] (1.5,0) arc
        [
            start angle=30,
            end angle=-110,
            x radius=1cm,
            y radius =.1cm
        ] ;
    \end{tikzpicture}
    }
}

\newcommand{\fourvloop}[1]{
\adjustbox{valign=c}{
    \begin{tikzpicture}
        \draw[black, thick] (0,-0.3) -- (2.5,-0.3);
        \draw[black, thick] (0,.5) -- (0.5,.5);
        \draw[black, thick] (2,.5) -- (2.5,.5);
        \draw[black, thick] (1,0) -- (1.5,0);
        \draw[black, thick] (0.5,.5) -- (1,0);
        \draw[black, thick] (1.5,0) -- (2,.5);
        \draw[black] (1.3,0.2) circle (0pt) node[anchor=center]{#1};
        \draw[black, thick] (1,0) arc
        [
            start angle=150,
            end angle=400,
            x radius=0.3cm,
            y radius =.1cm
        ] ;
    \end{tikzpicture}
    }
}

\newcommand{\fourvrepartII}[1]{
\adjustbox{valign=c}{
    \begin{tikzpicture}
        \draw[black, thick] (0,.5) -- (0.5,.5);
        \draw[black, thick] (0,0) -- (2.5,0);
        \draw[black, thick] (2,.5) -- (2.5,.5);
        \draw[black, thick] (1,0) -- (1.5,0);
        \draw[black, thick] (0.5,.5) -- (1,0);
        \draw[black, thick] (1.5,0) -- (2,.5);
        \draw[black] (1.3,0.2) circle (0pt) node[anchor=center]{#1};
        \draw[black, thick] (2.5,0) -- (3,.5);
        \draw[black, thick] (2.5,0.5) -- (3,0);
        \draw[black, thick] (3,0.5) -- (3.5,.5);
        \draw[black, thick] (0,0) arc
        [
            start angle=150,
            end angle=332,
            x radius=2cm,
            y radius =.2cm
        ] ;
        \draw[black, thick] (3,0) arc
        [
            start angle=30,
            end angle=-150,
            x radius=2cm,
            y radius =.2cm
        ] ;
        \draw[black, thick] (3.5,0.5) -- (4,0);
        \draw[black, thick] (3.5,-0.2) -- (4,.5);
    \end{tikzpicture}
    }
}

\newcommand{\sixvrepart}[2]{
\adjustbox{valign=c}{
    \begin{tikzpicture}
        \draw[black, thick] (0,.5) -- (0.5,.5);
        \draw[black, thick] (0,.25) -- (0.75,.25);
        \draw[black, thick] (2,.5) -- (2.5,.5);
        \draw[black, thick] (1,0) -- (1.5,0);
        \draw[black, thick] (0.5,.5) -- (1,0);
        \draw[black, thick] (1.5,0) -- (2,.5);
        \draw[black, thick] (1.75,.25) -- (2.5,.25);
        \draw[black] (1,0.3) circle (0pt) node[anchor=center]{#1};
        \draw[black] (1.5,0.3) circle (0pt) node[anchor=center]{#2};
        \draw[black, thick] (1,0) arc
        [
            start angle=150,
            end angle=290,
            x radius=1cm,
            y radius =.1cm
        ] ;
        \draw[black, thick] (1.5,0) arc
        [
            start angle=30,
            end angle=-110,
            x radius=1cm,
            y radius =.1cm
        ] ;
    \end{tikzpicture}
    }
}

\newcommand{\binorX}[4]{
\adjustbox{valign=c}{
    \begin{tikzpicture}
        \draw[black, thick] (0,0) -- (0.5,.5);
        \draw[black, thick] (0,0.5) -- (0.5,0);
        \draw[black] (-0.2,0) circle (0pt) node[anchor=center]{#1};
        \draw[black] (-0.2,0.5) circle (0pt) node[anchor=center]{#2};
        \draw[black] (0.7,0) circle (0pt) node[anchor=center]{#3};
        \draw[black] (0.7,0.5) circle (0pt) node[anchor=center]{#4};
    \end{tikzpicture}
    }
}

\newcommand{\binorII}{
\adjustbox{valign=c}{
    \begin{tikzpicture}
        \draw[black, thick] (0,0) -- (0.5,0);
        \draw[black, thick] (0,0.3) -- (0.5,0.3);
    \end{tikzpicture}
    }
}

\newcommand{\binorU}{
\adjustbox{valign=c}{
    \begin{tikzpicture}
        \draw[black, thick] (-0.25,0) arc
        [
            start angle=90,
            end angle=-90,
            x radius=.2cm,
            y radius =.2cm
        ] ;
        \draw[black, thick] (0.25, 0) arc
        [
            start angle=90,
            end angle=270,
            x radius=.2cm,
            y radius =.2cm
        ] ;
    \end{tikzpicture}
    }
}

\newcommand{\binorI}[2]{
\adjustbox{valign=c}{
    \begin{tikzpicture}
        \draw[black] (-0.2,0) circle (0pt) node[anchor=center]{#1};
        \draw[black] (0.7,0) circle (0pt) node[anchor=center]{#2};
        \draw[black, thick] (0,0) -- (0.5,0);
    \end{tikzpicture}
    }
}

\newcommand{\binorIcenter}[1]{
\adjustbox{valign=c}{
    \begin{tikzpicture}
        \draw[black] (0.5,0.3) circle (0pt) node[anchor=center]{#1};
        \draw[black, thick] (0,0) -- (1,0);
    \end{tikzpicture}
    }
}

\newcommand{\binorIsym}[1]{
\adjustbox{valign=c}{
    \begin{tikzpicture}
        \draw[black] (1,0) circle (0pt) node[anchor=center]{#1};
        \draw[black] (-0.2,0) circle (0pt) node[anchor=center]{#1};
        
        \draw[black, thick] (0,0.2) -- (0.3,0.2);
        \draw[black, thick] (0,0.4) -- (0.3,0.4);
        \draw[black, thick] (0,-0.4) -- (0.3,-0.4);
        \draw[black] (0.15,0) circle (0pt) node[anchor=center]{$\vdots$};
        
        \draw[black, thick] (0.3,-.5) rectangle ++(0.2,1);
        
        \draw[black, thick] (0.5,0.2) -- (0.8,0.2);
        \draw[black, thick] (0.5,0.4) -- (0.8,0.4);
        \draw[black, thick] (0.5,-0.4) -- (0.8,-0.4);
        \draw[black] (0.65,0) circle (0pt) node[anchor=center]{$\vdots$};
    \end{tikzpicture}
    }
}

\newcommand{\binorC}[2]{
\adjustbox{valign=c}{
    \begin{tikzpicture}
        \draw[black] (0.4,0) circle (0pt) node[anchor=center]{#1};
        \draw[black] (0.4,-0.4) circle (0pt) node[anchor=center]{#2};
        \draw[black, thick] (0.25, 0) arc
        [
            start angle=90,
            end angle=270,
            x radius=.2cm,
            y radius =.2cm
        ] ;
    \end{tikzpicture}
    }
}

\newcommand{\binorD}[2]{
\adjustbox{valign=c}{
    \begin{tikzpicture}
        \draw[black] (-0.5,0) circle (0pt) node[anchor=center]{#1};
        \draw[black] (-0.5,-0.4) circle (0pt) node[anchor=center]{#2};
        \draw[black, thick] (-0.25,0) arc
        [
            start angle=90,
            end angle=-90,
            x radius=.2cm,
            y radius =.2cm
        ] ;
    \end{tikzpicture}
    }
}

\newcommand{\threeval}[3]{
\adjustbox{valign=c}{
    \begin{tikzpicture}
        \draw[black] (-0.2,0) circle (0pt) node[anchor=center]{#1};
        \draw[black] (0.7,0.5) circle (0pt) node[anchor=center]{#2};
        \draw[black] (0.7,-0.5) circle (0pt) node[anchor=center]{#3};
        \draw[black, thick] (0,0) -- (0.5,0);
        \draw[black, thick] (0.5,0) -- (1,0.5);
        \draw[black, thick] (0.5,0) -- (1,-0.5);
        \draw[black, fill=black] (0.5,0) circle (2pt);
        
    \end{tikzpicture}
    }
}

\newcommand{\binorThreeVal}[6]{
\adjustbox{valign=c}{
    \begin{tikzpicture}
        \draw[black] (-0.75,0) circle (0pt) node[anchor=center]{$\vdots$};
        \draw[black, thick] (-1,0.2) -- (-0.5,0.2);
        \draw[black, thick] (-1,0.4) -- (-0.5,0.4);
        \draw[black, thick] (-1,-0.4) -- (-0.5,-0.4);
        \draw[black, thick] (-0.5,-.5) rectangle ++(0.2,1);
        
        \draw[black, thick, rotate around={60:(0,0)}] (0.5,-0.5) rectangle ++(0.2,1);
        \draw[black, thick, rotate around = {60:(0,0)}] (0.7,0.2) -- (1.2,0.2);
        \draw[black, thick, rotate around = {60:(0,0)}] (0.7,0.4) -- (1.2,0.4);
        \draw[black, thick, rotate around = {60:(0,0)}] (0.7,-0.4) -- (1.2,-0.4);
        \draw[black, rotate around = {60:(0,0)}] (0.95,0) circle (0pt) node[anchor=center, rotate around = {60:(0,0)}]{$\vdots$};
        
        \draw[black, thick, rotate around={-60:(0,0)}] (0.5,-0.5) rectangle ++(0.2,1);
        \draw[black, thick, rotate around = {-60:(0,0)}] (0.7,0.2) -- (1.2,0.2);
        \draw[black, thick, rotate around = {-60:(0,0)}] (0.7,0.4) -- (1.2,0.4);
        \draw[black, thick, rotate around = {-60:(0,0)}] (0.7,-0.4) -- (1.2,-0.4);
        \draw[black, rotate around = {-60:(0,0)}] (0.95,0) circle (0pt) node[anchor=center, rotate around = {-60:(0,0)}]{$\vdots$};
        
        \draw[black, thick] (-0.275, 0.3) arc
        [
            start angle=-90,
            end angle=-30,
            x radius=.4cm,
            y radius =.4cm
        ] ;
        \draw[black, thick] (-0.275, -0.3) arc
        [
            start angle=90,
            end angle=30,
            x radius=.4cm,
            y radius =.4cm
        ] ;
        \draw[black, thick] (0.525, 0.3) arc
        [
            start angle=120,
            end angle=220,
            x radius=.4cm,
            y radius =.4cm
        ] ;
        
        \draw[black] (-1.1,0) circle (0pt) node[anchor=center]{#1};
        \draw[black] (0.6,1.1) circle (0pt) node[anchor=center]{#2};
        \draw[black] (0.6,-1.1) circle (0pt) node[anchor=center]{#3};
        \draw[black] (-0.05,0.15) circle (0pt) node[anchor=center]{#4};
        \draw[black] (-0.05,-0.15) circle (0pt) node[anchor=center]{#5};
        \draw[black] (0.2,0) circle (0pt) node[anchor=center]{#6};
    \end{tikzpicture}
    }
}

\newcommand{\vstatefive}[7]{
\adjustbox{valign=c}{
    \begin{tikzpicture}
        \draw[black, thick] (0,0) -- (3,0);
        \draw[black, thick] (0,.25) -- (1.13,.25);
        \draw[black, thick] (1.87,.25) -- (3,.25);
        \draw[black, thick] (2,.5) -- (3,.5);
        \draw[black, thick] (1.13,.25) -- (1.3,0);
        \draw[black, thick] (1.7,0) -- (2,.5);
        \draw[black] (1.5,-0.2) circle (0pt) node[anchor=center]{#1};
        \draw[black] (1.65,0.2) circle (0pt) node[anchor=center]{#2};
        \draw[black] (0,0.35) circle (0pt) node[anchor=center]{\tiny #3};
        \draw[black] (0,.1) circle (0pt) node[anchor=center]{\tiny #4};
        \draw[black] (3,0.1) circle (0pt) node[anchor=center]{\tiny #5};
        \draw[black] (3,0.35) circle (0pt) node[anchor=center]{\tiny #6};
        \draw[black] (3,0.6) circle (0pt) node[anchor=center]{\tiny #7};
    \end{tikzpicture}
    }
}

%% file: Chapters/Introduction.tex
Perfect tensors are powerful structures that encode highly entangled states of many-body quantum systems. Their defining property is that any reduced density matrix of at most half of the system is maximally mixed. As such they describe \emph{absolutely maximally entangled} states \cite{Helwig:2012nha}. 
Perfect tensors have been studied in a broad spectrum of fields. A system of qubits, for instance, can be encoded by a perfect tensor with maximal code distance, in the context of quantum error correction. Perfect tensors have also been shown to \rema{reproduce features of} the AdS/CFT correspondence \cite{Pastawski_2015, Bhattacharyya16} and illustrate the Ryu-Takayanagi formula \cite{Ryu06} emerging from many-body entanglement. This way, erasure of quantum information can be reconstructed and a quantum error correcting code can be built \cite{Almheiri_2015, ChunJun22}. \rema{Random tensor networks with large bond dimension indicate similar features \cite{Hayden_2016}.}
What makes perfect tensors particularly interesting in the context of holographic duality is the fact that they represent quantum channels of maximal chaos \cite{Hosur_2016} making them suitable candidates for representing the holographic dual to the bulk quantum state of a black hole \cite{Maldacena_2016}. In all of these applications, perfect tensors appear as components of a tensor network.

From the perspective of quantum gravity, holography and the entanglement/geometry correspondence are two important open questions \cite{deBoer22}, and the models referred to contain very important insights. At the same time, they are somewhat ad hoc. It is therefore important to note that tensor networks also play an important role in the definition of quantum states in a quantum theory of gravity derived from first principles. In loop quantum gravity \cite{Thiemann:2007pyv,Ashtekar:2011ni}, special tensor networks represent basis states that diagonalize certain aspects of spatial geometry. As such, the tensors of these \emph{spin networks}  \cite{penrose1,Baez:1994hx} represent quantum states of geometric polyhedra \cite{Bianchi:2010gc}. 

\rema{The bond dimensions in spin networks are not necessarily uniform. They corresponds to geometric areas of the faces of polyhedra dual to the network. While we will mostly consider tensors with physical dimension 2 -- corresponding to spin 1/2 -- in the present work, higher and possibly non-uniform dimensions would also be physically relevant.\footnote{\rema{While in the path integral approach to LQG, the limit of large dimensions is often associated with the classical limit \cite{Barrett:2009mw}, this is not unequivocal. In the canonical approach, low spins play an important role. They make up the bulk of the states in the state counting for black hole entropy \cite{Ashtekar:2000eq}, for example, with large spins exponentially suppressed \cite{Domagala:2004jt}.}}}

Spin networks and their linear combinations have been studied under the aspect of entanglement entropy \cite{Bianchi:2018fmq} and holography \cite{Han_2017, Colafranceschi_2022, Colafranceschi22}. It has been shown in \cite{Han_2017} that the area operator of loop quantum gravity can be related to holographic entanglement entropy, thus recovering the Ryu-Takayanagi formula. 

The limits to the existence of perfect tensors have been investigated in the absence of additional symmetries \cite{Scott04, HIGUCHI2000213, GISIN19981} (see also \cite{Huber_2018, Huber_table} for an illustrative overview). Spin network states are more restrictive in a sense that they come with additional constraints on the tensor network states describing quantum geometry. 
The tensors of the spin networks have to be invariant under an action of SU(2) on all tensor factors. In the geometric picture, this invariance corresponds to the closure of the quantized polyhedra. Hence, a pressing question is the existence of invariant perfect tensors (IPTs). This would allow to study geometric correlates of perfectness and give the models for AdS/CFT a geometric interpretation using more than the the network topology. Previous work has investigated IPTs  for qudits (arbitrary local dimension $d$) and showed that every 3-valent Wigner 3j-symbol is an IPT \cite{Li_2017}. Apart from that, no more IPTs have been found. In fact, it has been shown that there are \rema{no 4-valent IPTs \cite{HIGUCHI2000213,Li_2017}}, but for valence $n \geq 5$ random invariant tensors are concentrated around perfect tensors \cite{Li_2018}.

In this work, we introduce a systematic method to find IPTs using a decomposition into special coupling schemes that we call bridge states. For the derivation we make use of Penrose's binor calculus \cite{Penrose71, Kauffman02}, a graphical notation which is particularly useful when dealing with spin network states \cite{Pietri_1997, Mansuroglu_2021}. We apply the proposed algorithm to give an additional proof for non-existence of 4-valent IPTs for qubit encoding, thereby reproducing a result of \cite{Li_2017}. We also apply the algorithm to 6-valent qubit encodings and show that also those hexagonal invariant perfect structures are forbidden. We conclude with a no-go theorem for arbitrary qubit encodings making use of the Scott bound \cite{Scott04}. Moreover, we put general constraints on the layout of IPTs with even valence: All indices are restricted to correspond to the same spin representation. A weaker restriction also holds for odd valence.

We structure the paper as follows. After defining the important notions, we use representation theory arguments to put restrictions to the layout of general IPTs and use that to prove (again, see  \cite{Li_2017}) a no-go theorem for 4-valent qubit in section \ref{sec:limits}. In section \ref{sec:main_master}, we derive a master equation for the coefficients determining an IPT of given valence and solve it for valence four and six in section \ref{sec:nogo}. With this, we present a second proof for the no-go theorem of 4-valent qubit encodings and a new no-go theorem of 6-valent qubit encodings which imply a no-go theorem for generic valence.

\rema{\emph{Note added in proof}: After publication of this work as a preprint and submission to JHEP, an independent work that proves the non-existence of perfect SU(2) invariant qubit codes \cite{Bernards_PHD} has been brought to our attention. Meanwhile it is available as a pre-print \cite{Bernards22}. The work rules out the perfectness of Werner states, i.e., qubit codes that are $U(d)$ invariant up to phase, where $d$ is the local dimension. $U(d)$ invariance up to phase is a stronger assumption for higher $d$, but for qubits their statements coincide with our finding that there are no SU(2) invariant perfect codes.}

%% file: Chapters/Basic_Def.tex
Let us start with the introduction of conventions and notions we need for the further study. We freely switch between considering the same object $I$ as a linear map $I:V_1\otimes \ldots \otimes V_k\rightarrow V_{k+1}\otimes\ldots\otimes V_n$, a state in $V_1^*\otimes \ldots \otimes V^*_k\otimes V_{k+1}\otimes\ldots\otimes V_n$ and its tensor representation $I^{a_1\ldots a_m}{}_{a_{k+1}\ldots a_n}$. We will furthermore not distinguish between $V$ and $V^*$ and pull indices with $\delta$.

All vector spaces will carry representations of $SU(2)$. We denote by $j$ both, a spin label $j\in \mathbb{N}/2$ and the corresponding $2j+1$-dimensional irrep of $SU(2)$.

The first important notion is that of an invariant tensor, also denoted intertwiner or equivariant map in the literature. 
\begin{defi}
    A tensor $I: j_1 \otimes ... \otimes j_k \to j_{k+1} \otimes ... \otimes j_{n}$ is called invariant, if it satisfies
    \begin{align}
        \tensor{g}{_{a_1}^{b_1}} \cdots \tensor{g}{_{a_k}^{b_k}} &\tensor{I}{^{a_1 ... a_k}_{a_{k+1} ... a_n}} \tensor{g}{^{a_{k+1}}_{b_{k+1}}} \cdots \tensor{g}{^{a_n}_{b_n}} \nonumber \\
        = &\tensor{I}{^{b_1 ... b_k}_{b_{k+1} ... b_n}}
    \end{align}
    for $g \in \textrm{SU(2)}$.
\end{defi}
Note that invariant tensors $I$ exist if and only if a coupling scheme to the trivial representation $j=0$ exists, i.e. $I: j_1 \otimes ... \otimes j_n \to 0$. Next, we define the notion of bipartition of tensors.
\begin{defi}
    Let $I$ be any tensor with coefficients $\tensor{I}{^{a_1 ... a_k}_{b_{k+1} ... b_n}}$ and $\mathcal{I} = \{1, ..., n\}$ its index set. A bipartition is a pair $(A, B)$ of ordered tuples $A, B \subseteq \mathcal{I}$, such that as sets $A\dot{\cup} B = \mathcal{I}$. 
    
    Compositions of two tensors $I, J$ with bipartitions $(A, B)$ and $(B, C)$, respectively, are understood in the following way
    \begin{align}
        \tensor{(JI)}{^{a_1 ... a_k}_{c_{k+1} ... c_n}} = \tensor{I}{^{a_1 ... a_k}_{b_{k+1} ... b_n}} \tensor{J}{^{b_1 ... b_k}_{c_{k+1} ... c_n}}
    \end{align}
\end{defi}
Hence, a bipartition can be thought of as the distinction between input and output of the tensor interpreted as a multilinear function. 
If a tensor is a partial isometry in all possible bipartitions that allow an embedding from input to output, i.e. $|A| \leq |B|$, or formally:
\begin{defi}
    A tensor $I: j_1 \otimes ... \otimes j_n \rightarrow 0$ is called perfect, if for any bipartition $(A, B)$ with $|A| \leq |B|$ there exists $\lambda \in \mathds{R}$, s.t. $I^\dagger I = \lambda \mathds{1}$.
\end{defi}

%% file: Chapters/Limitations_on_spin.tex
The requirement of perfectness puts a number of restrictions on invariant tensors. For every bipartition to be a partial isometry, we can study the action of an invariant tensor on irreducible subspaces of SU(2) using Schur's Lemma repeatedly.

\begin{prop}
\label{prop:irred_perfect}
    Let $I: j_1 \otimes ... \otimes j_n \rightarrow 0$ be a non-trivial, $n$-valent invariant, perfect tensor. For any bipartition $(A, B)$ with $|A| \leq |B|$, the following statements hold
    \begin{enumerate}
        \item $I$ maps irreducible spin $j$ subspaces to spin $j$ subspaces.
        \item $I\rvert_{j} = \mu_{j} \mathds{1} $ with $\mu_{j} \in \mathds{C}$ and $|\mu_{j}|^2 = |\mu_{j'}|^2 = \lambda(A) \quad \forall j, j'$
    \end{enumerate}
    These statements are depicted in the following diagram that shows the mapping of $I$ reduced to the irreducible subspaces.
    
    \begin{align}
        \begin{matrix}
            \sum_{i \in A} j_i & \xrightarrow{\mu_{j_{max}}} & \sum_{i \in B} j_i \\
            \sum_{i \in A} j_i - 1 
            & \xrightarrow{\mu_{j_{max} - 1}}
            & \sum_{i \in B} j_i - 1 \\
            & \qquad\qquad &\\
            ... && ... \\
            &&\\
            \min_{\vec k} |\sum_{i \in A} (-1)^{k_i} j_i | & \xrightarrow{\mu_{j_{min}}} & \min_{\vec k} |\sum_{i \in A} (-1)^{k_i} j_i |
        \end{matrix}
    \end{align}
    
    \begin{proof}
        \begin{enumerate}
            \item This is a direct consequence of Schur's lemma.
            \item Also by Schur's lemma, the action of $I$ reduced to an irreducible subrepresentation labelled by a spin quantum number $j$ is equal to $\mu_j \mathds{1}$ with $\mu_j \in \mathds{C}$. 
            
            For perfectness, we need $I^\dagger I = \lambda(A) \mathds{1}$ for every bipartition $\mathcal{I} = A \uplus B$, which translated into $|\mu_j|^2 = |\mu_{j'}|^2 =: \lambda \, \forall j, j'$ and for any bipartition $\mathcal{I} = A \uplus B$.
        \end{enumerate}
    \end{proof}
\end{prop}

From the $SU(2)$ invariance of $I$, we have assured that every spin $j$ subspace in the domain is also found in the image of $I$ as long as $|A| \leq |B|$. If we consider the case of even valence $2n$ and the balanced bipartitions, i.e. $|A| = |B|$, we can derive limitations on the spin quantum number.
\begin{prop}
\label{prop:even_n}
    Let $I$ be an invariant, perfect tensor with even valence $2n$. Then all ingoing spins are the same, i.e. $I : j \otimes ... \otimes j \to 0$
    \begin{proof}
        We will apply Lemma \ref{prop:irred_perfect} repeatedly for all balanced bipartitions, i.e. $|A| = |B|$. Since, we need to map the maximally coupled spin onto a subspace of equal dimension, we can infer the inequality 
        \begin{align}
            \sum_{i \in A} j_i \leq \sum_{i \in B} j_i,
        \end{align}
        for a fixed bipartition $A, B$. As we consider a balanced bipartition, we can also look at the bipartition, where $A$ and $B$ are swapped. This makes the inequality an equality
        \begin{align}
            \sum_{i \in A} j_i = \sum_{i \in B} j_i.
            \label{eq:irred_jmax}
        \end{align}
        Additionally, we can permute indices through $A$ and $B$ and get the same equation \eqref{eq:irred_jmax}. One can easily show by induction that \eqref{eq:irred_jmax} implies $j_i = j_k \forall i, k \in \mathcal{I}$.
    \end{proof}
\end{prop}

We can also make a statement about the spin quantum numbers of invariant, perfect tensors of odd valence $2n+1$ which is, however, less restrictive.
\begin{prop}
\label{prop:odd_n}
    Let $I$ be an invariant, perfect tensor with odd valence $2n+1$. Then the spin quantum numbers of the ingoing edges fulfill
    \begin{align}
        \sum_{i =1}^n j_i \leq \sum_{i =n+1}^{2n+1} j_i 
        \label{eq:odd_n}
    \end{align}
    where we sorted the spins in decreasing order, i.e. $j_1 \geq ... \geq j_{2n+1}$ for all bipartitions $\mathcal{I} = A \uplus B$, with $|A| = n$ and $|B| = n+1$.
    \begin{proof}
        The proposition is a direct application of Lemma \ref{prop:irred_perfect} with the bipartition where $A$ contains the $n$ largest spins.
    \end{proof}
\end{prop}
Note that \eqref{eq:odd_n} also implies inequalities of the same form but with arbitrary permutations of indices in $A$ and $B$, as long as $|A| +1 = |B|$. We conclude the section with a number of examples and remarks.

\begin{exa}
\begin{enumerate}
    \item The invariant tensor $I: \frac{1}{2} \otimes \frac{1}{2} \otimes \frac{1}{2} \otimes \frac{1}{2} \otimes 2 \to 0$ cannot be perfect by Lemma \ref{prop:odd_n}.
    \item The HaPPY code \cite{Pastawski_2015} is constructed from a perfect tensor with valence $n=6$. However, it is not invariant.
    \item If an invariant, perfect tensor encodes qubits only ($j_i = \frac{1}{2} \, \forall i$), it needs to have even valence. 
    
    Vice versa if an invariant, perfect tensor has even valence and shall encode at least one qubit, i.e. $j_i = \frac{1}{2}$ for one $i$, then $j_k = \frac{1}{2} \, \forall k$.
\end{enumerate}
\end{exa}
We want to apply Lemma \ref{prop:even_n} to the special case $n = 4$ and $j= \frac{1}{2}$ and reproduce a result of \cite{Li_2017} with a different proof.
\begin{theo}
\label{theo:no_4v_I}
    There are no non-trivial invariant, perfect tensors $I$ with valence $n=4$ and $j=\frac{1}{2}$.
\begin{proof}[Proof by contradiction (Sketch)]
    By Lemma \ref{prop:irred_perfect}, we know that the eigenvalues of an invariant, perfect tensor $I$ can only differ by a phase, i.e. $I$ takes the form
    \begin{align}
        &I^{abcd} = \sqrt{\lambda} (e^{i \phi_s} \epsilon^{a(c} \epsilon^{d)b} + e^{i \phi_a} \epsilon^{a[c} \epsilon^{d]b})
        \label{eq:no_4v_I}
    \end{align}
    for the bipartition $(ab) \to (cd)$. Repeating this argument for the repartitions $(ad) \to (cb)$ and $(ac) \to (bd)$ gives us similar expressions for $I$ with potentially different phases. Comparing the three versions of \eqref{eq:no_4v_I} yields two contradicting constraints for the phase difference $\Delta \phi = \phi_a - \phi_s$.
\end{proof}
\end{theo}
See appendix \ref{sec:no_4v_irred} for the full proof. While the arguments of Theorem \ref{theo:no_4v_I} can be generalized to higher valences and spins, it quickly gets combinatorically complicated. We will present a simpler systematic approach to get sufficient conditions for invariant perfect tensors.

%% file: Chapters/sufficient_cond.tex
A more down to earth approach to formulate sufficient conditions for an invariant, perfect tensor is achieved by decomposing tensor networks into \textit{bridge states}, coupling schemes\footnote{For more information on SU(2) coupling schemes see for example \cite{Wigner1993,Brunnemann:2004xi}.} that additionally contain the information of a bipartition. The following discussion is focused on a tensor of even valence and $j=\frac{1}{2}$, but it can be readily generalized to higher spins and odd valence. The only change is increase of combinatorial complexity \rema{(see Remark \ref{rem:generalize} for a sketch)}. Let us start by defining bridge states
\begin{defi}[Bridge State]
    \label{def:bridge}
    Let $n\in \mathds{N}$ and $(j_a)_{a \in \{1, ..., n_1\}}, (k_b)_{b \in \{1, ..., n_2\}}$ lists of spin quantum numbers where $n_1+n_2 = n$ and $j_{n_1} = k_{n_2}$. We define the SU(2) invariant bridge state associated with $j$ and $k$ as the following coupling scheme
    \begin{align}
        \label{eq:bridge_state}
        \curlyvee_{(j, k)} = \frac{1}{2}^{\otimes n_1}\vstate{j}{k}\frac{1}{2}^{\otimes n_2}
    \end{align}
\end{defi}
Since an invariant coupling exists only for certain data $n_1, n_2, j$ and $k$, bridge states exist only for certain spin data as well. Implicit in \eqref{eq:bridge_state} is a certain normalisation. For details see the Appendix \ref{sec:master}.

We can directly infer a number of helpful relations between bridge states and invariant tensors 
\begin{rem}
\begin{enumerate}
    \item The bridge state $\curlyvee_{(j, k)}$ exists, if and only if $k_1 = j_1 = \frac{1}{2}$ and $j_a = |j_{a-1} \pm \frac{1}{2}| \, \forall a$, and analogously for $k_a$
    \item Every bridge state implicitly encodes a bipartition of indices. For valence $2n$ and $n_1 = n_2 = n$, for instance, we get the balanced bipartition (cf. section \ref{sec:limits})
    \item \rema{As every bridge state is also a coupling scheme \cite{Wigner1993}, the set of invariant bridge states for fixed valence and bipartition $n_1, n_2$ represents a basis on the space of invariant tensors. }
\end{enumerate}
\end{rem}
With this, we can decompose any invariant tensor $I$ into bridge states $\curlyvee_{(j,k)}$ in which a bipartition is implicitly fixed in the ordering of $(j,k)$
\begin{align}
    I = \sum_{\substack{j,k \\ j_{n_1} = k_{n_2}}} c_{j,k} \curlyvee_{(j,k)}.
    \label{eq:I_dec}
\end{align}
In the following, we will drop the addendum $j_{n_1} = k_{n_2}$ and assume the coefficients $c_{j,k}$ to be zero if this is not satisfied. To derive a necessary and sufficient condition for perfectness of $I$, we need to translate the isometry relation
\begin{align}
    I^\dagger I = \lambda \mathds{1}
    \label{eq:perfectness}
\end{align}
into a set of equations for the coefficients $c_{j,k}$. To write down those equations, we still need the notion of theta networks
\begin{defi}[Theta Network]
\label{def:theta}
    Let $2n\in 2\mathds{N}$ 
    and $(j_a), (k_b), \, a,b \in \{1, ..., n\}$ lists of spin quantum numbers where $j_1 = k_1$ and $j_n = k_n$. The theta-network associated with $j$ and $k$ is defined as
    \begin{align}
        _j\Theta_k = \thetaNet{j}{k},
    \end{align}
    where the inner edges all carry spin $\frac{1}{2}$.
\end{defi}
Since the theta network has no open ends, it evaluates to a number, and we will denote this number as $_j\Theta_k$ in the following.
\begin{theo}
\label{theo:master}
    Let $I$ be an invariant tensor and $c_{j,k}$ coefficients denoting its bridge state decomposition for a fixed bipartition via \eqref{eq:I_dec}. $I$ is a perfect tensor, if and only if for every bipartition $A$ the coefficients satisfy
    \begin{align}
        \sum_k c_{j,k} \overline{c_{j',k}} \hspace{0.5em} {}_k\Theta_k = \lambda(A) \delta_{j,j'} \frac{(2j_{n_1} + 1)^2}{_j\Theta_j},
        \label{eq:master}
    \end{align}
    where $\lambda(A)\in \mathds{R}$ might be different for each bipartition $A$.
    \begin{proof}[Proof (Sketch)]
        The equations \eqref{eq:master} result from a comparison of coefficients of the two sides of \eqref{eq:perfectness}. The left hand side requires the contraction of $I$ with its own adjoint $I^\dagger$, which can be calculated using orthogonality relations of bridge states and identities from the binor calculus \cite{Kauffman02} (cf. appendix \ref{sec:master}). For the right hand side, a decomposition of the identity operator $\mathds{1}$ into bridge states is done (cf. Lemma \ref{prop:1_dec}). 
    \end{proof}
\end{theo}
For a detailed proof, see appendix \ref{sec:master}. We can directly solve for one of the coefficients from \eqref{eq:master} for general valence.
\begin{coro}
\label{cor:max_coeff}
    Let $I$ be an invariant partial isometry of valence $2n$ and $c_{j,k}$ the coefficients of its bridge state decomposition. It holds
    \begin{align}
        |c_{j_{max}, j_{max}}|^2 = \lambda,
    \end{align}
    where $j_{max} = \left( \frac{1}{2}, 1, ..., \frac{n}{2} \right)$.
    \begin{proof}
        This is a direct consequence of \eqref{eq:master}, when choosing $j=j'=j_{max}$.
        \begin{align}
            \sum_k c_{j_{max},k} \overline{c_{j_{max},k}} \hspace{0.5em} {}_k\Theta_k = \lambda \frac{(2j_{n} + 1)^2}{_{j_{max}}\Theta_{j_{max}}},
        \end{align}
        In order for $I$ to be invariant, only one addend with $k = j_{max}$ remains on the left hand side. Finally, using $_{j_{max}}\Theta_{j_{max}} = 2j_n + 1$ yields the statement.
    \end{proof}
\end{coro}
Since \eqref{eq:master} has to hold for every possible bipartition, it remains to study how these equations transform under a change of bipartition, a repartition. The repartitioned tensor $\tilde I$ also can be decomposed into bridge states with new coefficients $\tilde c_{j,k}$ that, again, have to satisfy \eqref{eq:master}. 

However, the master equation \eqref{eq:master} only needs to be solved once. Since the new coefficients of the repartitioned tensor can be related to the old coefficients by a linear operator $R$
\begin{align}
    \tilde c_{(j',k')} = \tensor{R}{_{(j',k')}^{(j,k)}} c_{(j,k)},
\end{align}
it only remains to check consistency of the solution of the master equation \eqref{eq:master} under every such transformation. In particular, every repartition $R$ that does not change the balance between the index sets $A, B$ can be written as a series of $r$ permutations of neighboring indices
\begin{align}
    R = P_{i_r, i_r+1} \cdots P_{i_1, i_1 + 1}.
\end{align}
Note that the $\{P_{i, i+1}\}$ do not necessarily commute. To check \eqref{eq:master} for every possible bipartition hence requires the exploration of the whole algebra generated by $\{P_{i, i+1}\}$, in general. We show in appendix \ref{sec:repartitions} that some repartitions are trivial in the sense that, if $I$ is an IPT then also the repartitioned tensor $\tilde I$ is an IPT. In particular, every unbalancing repartition which moves an index from the domain into the image space is trivial (cf. Lemma \ref{prop:trvial_I}). Also balanced repartitions that do not mix indices of the domain or image are trivial (cf. Lemma \ref{prop:trvial_II}). 

It thus remains to investigate the permutation of the neighboring indices across the bridge $P^*$ (cf. Lemma \ref{prop:non_triv_repart})
\begin{align}
    \bipartstate{$I$}
\end{align}
and all combinations of trivial repartitions with $P^*$. We summarize the systematics in the following algorithm
\begin{alg}[Finding an IPT]
\label{alg:IPT}
    For a given valence $2n$ the set of IPTs encoding qubits ($j=\frac{1}{2}$) can be found the following way
    \begin{enumerate}
        \item Solve the master equations \eqref{eq:master} for one specific bipartition and parameterize the general solution into free amplitudes $A_1, A_2, ...$ and phases $\phi_1, \phi_2, ...$
        \item Calculate the basis transformation $P^*$ that yields the new coefficients $\tilde c$ corresponding to the repartitioned tensor which has the neighboring indices flipped across the bridge
        \item Parameterize the $\tilde c$ the same way as in 1. using new, independent parameters $\tilde A_1, \tilde A_2, ..., \tilde \phi_1, \tilde \phi_2, ...$
        \item Implement constraints on said parameters by identifying $\tilde c = P^* c$
        \item Repeat 2.~-- 4. with a repartition exchanging the last index in the domain with the second last index ($P_{n+1, n+2} P^* P_{n+1, n+2}$), third last index ($P_{n+2, n+3} P_{n+1, n+2} P^* P_{n+1, n+2} P_{n+2, n+3}$) etc.
        \item Repeat 2. -- 4. with the mirroring repartition that exchanges all domain and image indices
    \end{enumerate}
\end{alg}
In the end, Algorithm \ref{alg:IPT} requires at most $n+1$ repartitions. If the solution set becomes empty or only contains trivial solutions before every repartition has been checked, the algorithm terminates as no IPT exists.

Note that with Algorithm \ref{alg:IPT}, it is trivial to see that every 3-valent invariant tensor is also perfect. Since the tensor is unique there is no mixing of basis states. The perfectness relation can be calculated in two lines with the binor calculus (cf. Lemma \ref{prop:loopstate} for spin $\frac{1}{2}$ case).

%% file: Chapters/nogo.tex
We have already shown in Theorem \ref{theo:no_4v_I}, that there are no invariant perfect tensors with valence $2n=4$ and spin $j= \frac{1}{2}$. This can also be shown in the bridge state decomposition using Algorithm \ref{alg:IPT}. We present the full proof in appendix \ref{sec:no_4v} and sketch the most important arguments here.
\begin{proof}[Second Proof for Theorem \ref{theo:no_4v_I} (Sketch)]
    Assume we had found an IPT $I$. For a fixed bipartition, we can fix the coefficients of the bridge state decomposition of $I$ up to two phases (one relative phase)
    \begin{align}
        |c_0|^2 = \frac{\lambda}{4}, \qquad |c_1|^2 = \lambda.
    \end{align}
    While the mirroring repartition is trivial, the repartition $P^*$ fixes the relative phase to $\Delta \phi = 0$. The last repartition $P_{3, 4} P^* P_{3,4}$, however, would fix the relative phase of the initial coefficients to $\Delta \phi = \pi$, a contradiction.
\end{proof}

\rema{Note that 4-valent perfect tensors have been already ruled out without imposing SU(2) invariance in \cite{HIGUCHI2000213}.} One is tempted to go one step further and study the second smallest non-trivial tensor for spin $\frac{1}{2}$ particles which is a 6-valent tensor. We will show in an analogous way that also those hexagonal structures cannot be perfect and $SU(2)$ invariant at the same time. As before, a detailed proof is presented in appendix \ref{sec:no_6v} and a sketch of the proof is given here.
\begin{theo}
\label{theo:no_6v}
    There are no non-trivial invariant, perfect tensors $I$ with valence $2n=6$ and $j=\frac{1}{2}$.
    \begin{proof}[Proof (Sketch)]
        Assume we had found an IPT $I$. For a fixed bipartition, the master equation \eqref{eq:master} is solved leaving one free amplitude $A \in \mathds{R}^+$ and four phases $\phi, \rho, \chi, \psi \in [0,2\pi)$
        \begin{widetext}
            \begin{align}
                c_{(1,\frac{3}{2}, 1)} = \sqrt{\lambda} e^{i \phi} \quad
                c_{(1,\frac{1}{2}, 1)} = A e^{i \rho} \quad
                c_{(1,\frac{1}{2}, 0)} = \sqrt{\frac{\lambda}{3} - \frac{3}{4} A^2 } \, e^{i    \chi} \quad
                c_{(0,\frac{1}{2}, 0)} = \frac{3}{4} A e^{i\psi} \quad
                c_{(0,\frac{1}{2}, 1)} = \sqrt{\frac{\lambda}{3} - \frac{3}{4} A^2} e^{i(\rho + \psi - \chi + \pi)}
            \end{align}
        \end{widetext}
        The repartition $P^*$ reduces the degrees of freedom by fixing $A = \frac{2}{3} \sqrt{\lambda}$ and $\rho = \phi$. Furthermore, the phase $\chi$ does not appear anymore, since the coefficients of asymmetric bridge states are necessarily zero $c_{(0,\frac{1}{2}, 1)} = c_{(1,\frac{1}{2}, 0)} = 0$. Finally, the repartition acting on the second last index on the right hand side $P_{4, 5} P^* P_{4,5}$ projects the solution space to $\lambda = 0$ which is a trivial solution.
    \end{proof}
\end{theo}

Finally, we can use Theorem \ref{theo:no_6v} to close a gap and deduce a no-go theorem for invariant perfect qubit encodings. 
\begin{coro}
    There are no invariant perfect qubit ($j=\frac{1}{2}$) encodings.
    \begin{proof}
        The statement is a consequence of previous results. A 2-valent IPT is trivially and uniquely given by the identity $I = \mathds{1}$. From invariance, the valence is restricted to even $n$. Valence four is prohibited by Theorem \ref{theo:no_4v_I} and valence six by Theorem \ref{theo:no_6v}. Finally, the Scott bound for perfect tensors \cite{Scott04} gives the upper bound 
        \begin{align}
            n \leq 2 (d^2 - 1)
        \end{align}
        for the valence,
        with $d$ being the local index dimension. With $d = 2$ for qubits, we have $n \leq 6$ excluding any other valence.
    \end{proof}
\end{coro}

%% file: Chapters/Conclusion.tex
\change{This paper presents} a methodology that allows the systematic search for invariant perfect tensors. We \change{restrict} the spin quantum numbers of tensors with even and odd valence separately. While for even valent IPTs all indices must necessarily be in the same irreducible representation, for odd valence we \change{prove} a weaker condition \eqref{eq:odd_n} that restricts the spin quantum numbers to be not too far away from each other.

We further derive a master equation to determine the coefficients of an IPT decomposed into specific coupling schemes, the bridge states. Following our algorithm, we provide two proofs for the non-existence of 4-valent and one proof for 6-valent IPTs. Combining our results with the Scott bound for perfect tensors, we finally show that there can be no IPTs with qubit indices only.

Several generalizations of our methods to higher spin representations and odd valence leave us with open questions. While the generalization of the master equations \eqref{eq:master} are straightforward, finding solutions  quickly becomes more complex and technical. It thus remains to be shown how further general identities and symmetries can be translated into our framework and exploited to simplify Algorithm \ref{alg:IPT}.

Another open question that we leave for future work is the study of an alternative definition of IPTs. Instead of allowing at most half of the \textit{indices/particles} on the input side of any bipartition, one might ask for more and allow also bipartitions for which the dimension of the input representation space is smaller or equal to the output. This still allows for an isometric embedding as asked for in the definition of perfect tensors. Although, more bipartitions have to satisfy the isometry relation, the restrictions from section \ref{sec:limits} do not generalize straightforwardly. Moreover, if every index has the same local dimension $d$, the two definitions coincide. It remains to be shown how the alternative definitions relate, in general.

%% file: Chapters/no_4_v_irred.tex
\label{sec:no_4v_irred}
\begin{proof}[Proof of Theorem \ref{theo:no_4v_I}]
    We proof the statement by contradiction. If there was an invariant, perfect tensor $I$, we knew from Lemma \ref{prop:even_n} that all spin quantum numbers are equal. For a fixed balanced bipartition $|A| = |B|$, the tensor hence maps the following irreducible subrepresentations
    \begin{align}
        \begin{matrix}
            1 & \xrightarrow{\mu_{1}} & 1 \\
            0 & \xrightarrow{\mu_{0}} & 0
        \end{matrix}
    \end{align}
    where $|\mu_0| = |\mu_1|$. This gives a restriction to $I$
    \begin{align}
        &I^{abcd} = \sqrt{\lambda} (e^{i \phi_s} \epsilon^{a(c} \epsilon^{d)b} + e^{i \phi_a} \epsilon^{a[c} \epsilon^{d]b})
    \end{align}
    However, the same holds for every other balanced bipartition, for instance if we switch the indices $b$ and $d$
    \begin{align}
        I^{abcd} = \sqrt{\lambda} (e^{i \phi_s} \epsilon^{a(c} \epsilon^{d)b} + e^{i \phi_a} \epsilon^{a[c} \epsilon^{d]b}) = \sqrt{\lambda} (e^{i \psi_s} \epsilon^{a(c} \epsilon^{b)d} + e^{i \psi_a} \epsilon^{a[c} \epsilon^{b]d}),
        \label{eq:simple_bipart}
    \end{align}
    which yields constraints on the angles, when calculating the norm of $I$ 
    \begin{align}
        \braket{I|I} = &\lambda (e^{-i \phi_s} \epsilon^{a(c} \epsilon^{d)b} + e^{-i \phi_a} \epsilon^{a[c} \epsilon^{d]b}) (e^{i \psi_s} \epsilon_{a(c} \epsilon_{b)d} + e^{i \psi_a} \epsilon_{a[c} \epsilon_{b]d}) = \\
        =&\lambda \left( -\frac{3}{2}e^{i(\psi_s-\phi_s )} -\frac{3}{2} e^{i(\psi_a-\phi_s )} - \frac{3}{2} e^{i(\psi_s-\phi_a )} + \frac{1}{2}e^{i(\psi_a - \phi_a )} \right) \\
        =:&\lambda \left( -\frac{3}{2}e^{i\theta_s} -\frac{3}{2} e^{i(\theta_a + \theta_s - \xi)} - \frac{3}{2} e^{i\xi} + \frac{1}{2}e^{i\theta_a} \right)
        \stackrel{!}{=} 4\lambda = \braket{I|I},
        \label{eq:simple_bipart_angles}
    \end{align}
    where we evaluated the inner product of the two sides of \eqref{eq:simple_bipart} on the left hand side and the norm of $I$ on the right hand side using $\braket{I|I} = \Tr(I^\dagger I) = \lambda \Tr(\mathds{1}) = \lambda (2j+1)^d$. We already plugged in the spin quantum number $j=\frac{1}{2}$ and the image space dimension $d= \frac{n}{2} = 2$. If we split \eqref{eq:simple_bipart_angles} into real and imaginary part, we can read off the two conditions
    \begin{align}
        -\frac{3}{2}\cos{\theta_s} -\frac{3}{2} \cos{(\theta_a + \theta_s - \xi)} - \frac{3}{2} \cos{i\xi} + \frac{1}{2} \cos{\theta_a} &= 4 \label{eq:cos_walk} \\
        \frac{3}{2}\sin{\theta_s} +\frac{3}{2} \sin{(\theta_a + \theta_s - \xi)} + \frac{3}{2} \sin{\xi} - \frac{1}{2} \sin{\theta_a} &= 0
    \end{align}
    This equation has infinitely many solutions, for instance $\theta_a = \theta_s = \xi = \pi$. However, fixing $\xi$ also fixes the relative angle between the $\phi_{a/s}$, since $\xi = \psi_s - \phi_a = \theta_s + (\phi_s - \phi_a)$. We can read \eqref{eq:cos_walk} as a four-step walk in the complex plane with the destination $4 + 0i$. To reach the destination, we need the three terms with amplitude $\frac{3}{2}$ to be positive and $|\cos(...)| \geq \frac{7}{9}$. This implies the following conditions
    \begin{align}
        \pi -x < \theta_s < \pi + x \qquad \pi-x < \theta_s - \Delta \phi < \pi+x \qquad \pi-x < \theta_a + \Delta\phi < \pi+x
    \end{align}
    with $x = \cos^{-1}\left( \frac{7}{9} \right)$ and $\Delta \phi = \phi_a - \phi_s$. This implies
    \begin{align}
        \theta_s - \pi - x < \Delta \phi < \theta_s - \pi + x \qquad \text{ or } \Delta \phi \in [-2x, 2x]
        \label{eq:ciritical_cond}
    \end{align}
    
    If we were to chose a different balanced bipartition, like switching the indices $b$ and $c$, the can follow the same approach
    \begin{align}
        I^{abcd} = \sqrt{\lambda} (e^{i \phi_s} \epsilon^{a(c} \epsilon^{d)b} + e^{i \phi_a} \epsilon^{a[c} \epsilon^{d]b}) = \sqrt{\lambda} (e^{i \tilde \psi_s} \epsilon^{a(b} \epsilon^{d)c} + e^{i \tilde \psi_a} \epsilon^{a[b} \epsilon^{d]c}),
        \label{eq:simple_bipartII}
    \end{align}
    with new angles $\tilde \psi_{a/s}$. The new constraints on the angles, when calculating the norm of $I$ are slightly different
    \begin{align}
        \braket{I|I} = &\lambda (e^{-i \phi_s} \epsilon^{a(c} \epsilon^{d)b} + e^{-i \phi_a} \epsilon^{a[c} \epsilon^{d]b}) (e^{i \tilde \psi_s} \epsilon_{a(b} \epsilon_{d)c} + e^{i \tilde \psi_a} \epsilon_{a[b} \epsilon_{d]c}) = \\
        =&\lambda \left( -\frac{3}{2}e^{i(\tilde \psi_s-\phi_s )} +\frac{3}{2} e^{i(\tilde \psi_a-\phi_s )} + \frac{3}{2} e^{i(\tilde \psi_s-\phi_a )} + \frac{1}{2}e^{i(\tilde \psi_a - \phi_a )} \right) \\
        =:&\lambda \left( -\frac{3}{2}e^{i\tilde \theta_s} +\frac{3}{2} e^{i(\tilde \theta_a + \tilde \theta_s - \tilde \xi)} + \frac{3}{2} e^{i\tilde \xi} + \frac{1}{2}e^{i\tilde \theta_a} \right)
        \stackrel{!}{=} 4\lambda = \braket{I|I},
        \label{eq:simple_bipart_anglesII}
    \end{align}
    We get similar constraints from \eqref{eq:simple_bipart_anglesII}
    \begin{align}
        -\frac{3}{2}\cos{\tilde \theta_s} + \frac{3}{2} \cos{(\tilde \theta_a + \tilde \theta_s - \tilde \xi)} + \frac{3}{2} \cos{i\tilde \xi} + \frac{1}{2} \cos{\tilde \theta_a} &= 4 \\
        - \frac{3}{2}\sin{\tilde \theta_s} +\frac{3}{2} \sin{(\tilde \theta_a + \tilde \theta_s - \tilde \xi)} + \frac{3}{2} \sin{\tilde \xi} + \frac{1}{2} \sin{\tilde \theta_a} &= 0,
    \end{align}
    which again put conditions on the angles and consequently on $\Delta \phi$
    \begin{align}
        \pi -x < \tilde \theta_s < \pi + x \qquad -x < \tilde \theta_s - \Delta \phi < x \qquad -x < \tilde \theta_a + \Delta\phi < x \\
        \implies \tilde \theta_s -x < \Delta \phi < \tilde \theta_s + x \qquad \text{ or } \Delta \phi \in [\pi - 2x, \pi + 2x]
    \end{align}
    If we compare this constraint with \eqref{eq:ciritical_cond}, we realize a contradiction for the two ranges for $\Delta \phi$ have no overlap. This is due to 
    \begin{align}
        2x < \pi - 2x \iff x = \cos^{-1}\left( \frac{7}{9} \right) < \frac{\pi}{4} \iff \cos\left( \frac{\pi}{4} \right) = \frac{1}{\sqrt{2}} < \frac{7}{9}
    \end{align}
    which is true. Note that the last equivalence contains a flip of inequality sign due to the monotonously negative slope of the cosine. Analogously, $\pi + 2x < 2 \pi - 2x$ for the same reason.
\end{proof}

%% file: Chapters/thetaNet.tex
\label{sec:master}
In order to formulate a condition for perfectness, we \rema{introduce the binor calculus and present} a number of propositions that will turn out helpful. \rema{Tensor contractions can be depicted using the following basic definitions}
\begin{align}
    \binorI{A}{B} = \delta_A^B, \qquad \binorC{A}{B} = \epsilon^{AB}, \qquad \binorD{A}{B} = \epsilon_{AB}, \qquad \binorX{A}{B}{C}{D} = - \delta_A^D \delta_B^C
\end{align}
\rema{Higher spin $j$ representations are constructed from the spin $\frac{1}{2}$ representation by $2j$ symmetrized indices depicted by a rectangle. Every invariant tensor can be decomposed into a contraction of 3-valent invariant tensors. In the binor calculus these are depicted as follows \cite{Kauffman02, Pietri_1997}.}
\begin{align}
    \binorIcenter{$j$} = 
    \binorIsym{$j$},
    \qquad \qquad
    \threeval{$j$}{$k$}{$l$} = 
    \binorThreeVal{$j$}{$k$}{$l$}{$a$}{$b$}{$c$}
\end{align}

\rema{Here, the spins $a, b, c$ are uniquely chosen such that the number of in- and outflowing spin $\frac{1}{2}$ indices is preserved. In the following, we will deviate from a common notational convention by defining the arcs without an imaginary factor $i$.} 
This changes the binor identity to
\begin{align}
    \binorX{}{}{}{} = \binorU - \binorII.
\end{align}

\begin{prop}
\label{prop:loopstate}
    Let $j, k \in \frac{\mathds{N}}{2}$ be spin quantum numbers. Then the following tensor identity holds
    \begin{align}
        \loopstate{j}{k}{$\frac{1}{2}$}{$j \pm \frac{1}{2}$} = \delta_{j, k} \frac{2j + \frac{3}{2} \pm \frac{1}{2}}{2j + 1} \identity{j}
        \label{eq:loopresolve}
    \end{align}
    \begin{proof}
        The Kronecker delta $\delta_{j, k}$ arises from the fact that tensor must be invariant. Hence, we will assume $j=k$ from now on. Both cases $j \pm \frac{1}{2}$ can be proven with the binor calculus \cite{Kauffman02}. Let us start with the case $j - \frac{1}{2}$. The tensor takes the form
        \begin{align}
            \binorMin{j}{j}{$\frac{1}{2}$}{$j - \frac{1}{2}$} = \identity{j},
        \end{align}
        where we used the projective property of the symmetrizer. The case $j + \frac{1}{2}$ is slightly more involved
        \begin{align}
            \binorMax{j}{j}{$\frac{1}{2}$}{$j + \frac{1}{2}$} = \frac{1}{2j+1} \left( \tikz{ \draw[black, thick] (0,0) circle (.1)} + 2j \tikz{ \draw[black, thick] (0,0) circle (.1); \draw[black, thick] (-0.2,-0.1) -- (0.2, -0.1)} \right) \identity{j} = \frac{2j+2}{2j+1} \identity{j},
        \end{align}
        where we resolved the symmetrization over $j+\frac{1}{2}$ and merged the symmetrization over $j$ again, and finally used simple binor identities. 
    \end{proof}
\end{prop}

We can use Lemma \ref{prop:loopstate} to get a handy expression for the appearance of theta-network structures.
\begin{prop}
    Let $2n\in 2\mathds{N}$ and $j_a, k_b, \, a,b \in \{1, ..., n\}$ lists of spin quantum numbers with $j_n = k_n$. Here, we denote $j = (j_2, ..., j_n)$ and $k = (k_2, ..., k_n)$. Then
    \begin{align}
        \thetaLoop{$j_1$}{$k_1$}{j}{k} = \delta_{j, k} \delta_{j_1, k_1} \prod_{a=2}^n \left( \delta_{j_a, j_{a-1} + \frac{1}{2}} \frac{2j_{a-1} + 2}{2j_{a-1} + 1} + \delta_{j_a, j_{a-1} - \frac{1}{2}} \right) \identity{$j_1$}
        \label{eq:loopstate}
    \end{align}
    \begin{proof}
    This identity is a direct consequence of successively applying Lemma \ref{prop:loopstate}, reducing the loops one after another picking up a factor of $\frac{2j_{a-1}+2}{2j_{a-1} + 1}$ only in the case that $j_a = j_{a-1} + \frac{1}{2}$.
    \end{proof}
\end{prop}

\begin{coro}
\label{cor:theta}
    Let $2n\in 2\mathds{N}$ and $j_a, k_b, \, a,b \in \{1, ..., n\}$ lists of spin quantum numbers with $j_n = k_n$. Then
    \begin{align}
        _j\Theta_k = \delta_{j, k} \delta_{j_1, k_1} \prod_{a=2}^n \left( \delta_{j_a, j_{a-1} + \frac{1}{2}} \frac{2j_{a-1} + 2}{2j_{a-1} + 1} + \delta_{j_a, j_{a-1} - \frac{1}{2}} \right) (2j_1+1)
    \end{align}
    \begin{proof}
    The statement becomes obvious when taking the trace of \eqref{eq:loopstate}.
    \end{proof}
\end{coro}

The last ingredient we need is the decomposition of the identity operator into balanced bridge states $\curlyvee_{(j, k)}$. In general, we can decompose any tensor $I$ into tensors $\{V_{i}\}$ which are orthogonal with respect to the Hilbert Schmidt product
\begin{align}
    \braket{V_i | V_j} = \delta_{i, j} ||V_i||^2.
    \label{eq:HSprod}
\end{align}
In \eqref{eq:HSprod}, the tensors $V_i$ do not need to be normalized. We will show now that the bridge states as defined in Definition \ref{def:bridge} satisfy these orthogonality relations.

\begin{prop}
    Let $\curlyvee_{(j, k)}$ and $\curlyvee_{(j', k')}$ be two bridge states as defined in Definition \ref{def:bridge}. The bridge states are orthogonal, i.e.
    \begin{align}
        \braket{\curlyvee_{(j,k)}|\curlyvee_{(j', k')}} = \delta_{j,j'} \delta_{k,k'} ||\curlyvee_{(j,k)}||^2
    \end{align}
    with the norm being just a theta network as defined in Definition \ref{def:theta}
    \begin{align}
        ||\curlyvee_{(j,k)}||^2 = _{(j,k)}\Theta_{(j,k)}
    \end{align}
    \begin{proof}
        We calculate
        \begin{align}
            \braket{\curlyvee_{(j,k)}|\curlyvee_{(j', k')}} = \Tr \left( \curlyvee^\dagger_{(j,k)} \curlyvee_{(j',k')} \right)=\vnorm{k}{j}{j'}{k'} = _{(j, k)}\Theta_{(j', k')} = \delta_{j,j'} \delta_{k,k'} {}_{(j, k)}\Theta_{(j, k)},
        \end{align}
        where we indicated index contractions by dotted lines and the last step uses the orthogonality of theta networks from Corollary \ref{cor:theta}.
    \end{proof}
\end{prop}
\begin{rem}
As every bridge state is also a coupling scheme into Wigner 3j symbols \cite{Wigner1993} the set of all invariant $n$-valent bridge states is also a complete basis.
\end{rem}
The decomposition of the identity operator into bridge states will be of central importance in the following discussion.
\begin{prop}
\label{prop:1_dec}
    Let $\curlyvee_{(j,k)}$ be balanced ($n_1 = n_2 = n$) bridge states as defined in Definition \ref{def:bridge}. The bridge state decomposition of the identity reads
    \begin{align}
        \mathds{1} = \sum_{j,k} \frac{{}_j\Theta_{k}}{{}_{(j,k)}\Theta_{(j,k)}} \curlyvee_{(j,k)}
    \end{align}
    \begin{proof}
    Since the bridge states $\{\curlyvee_{(j,k)}\}$ form a basis of the invariant tensors, we can write the identity as a linear combination
    \begin{align}
        \mathds{1} = \sum_{j,k} c_{j,k} \curlyvee_{(j,k)}
    \end{align}
    The coefficients $c_{j,k}$ can be found by calculating the inner products
    \begin{align}
        c_{j,k} = \frac{\braket{\curlyvee_{(j,k)}|\mathds{1}}}{||\curlyvee_{(j,k)}||^2} = \frac{1}{{}_{(j,k)}\Theta_{(j,k)}} \Tr(\curlyvee^\dagger_{(j,k)}) = \frac{{}_j\Theta_{k}}{{}_{(j,k)}\Theta_{(j,k)}}
    \end{align}
    which follows from the orthogonality of the bridge states $\curlyvee_{(j,k)}$.
    \end{proof}
\end{prop}

With the preparations made, we can now precisely formulate conditions for invariant perfect tensors.

%% file: Chapters/master_eq.tex
\begin{proof}[Proof of Theorem \ref{theo:master}]
        We proceed by calculating both sides of \eqref{eq:perfectness} in the bridge state decomposition. As we already know the bridge state decomposition of $\mathds{1}$ from Lemma \ref{prop:1_dec}, we are left with calculating the left hand side:
        \begin{align}
            I^\dagger I &= \sum_{j,k,j',k'} c_{j,k} \overline{c_{j',k'}} \vdagv{j}{k}{k'}{j'} = \sum_{j,k,j',k'} c_{j,k} \overline{c_{j',k'}} \delta_{k,k'} \frac{_k\Theta_{k'}}{2k_{n_2}+1} \curlyvee_{(j,j')} \nonumber \\
            &= \sum_{j,k,j'} c_{j,k} \overline{c_{j',k}} \frac{_k\Theta_{k}}{2k_{n_2}+1} \curlyvee_{(j,j')}
            \label{eq:vdagv_dec}
        \end{align}
        For the right hand side of \eqref{eq:perfectness} let us recall the bridge decomposition of $\mathds{1}$
        \begin{align}
            \mathds{1} = \sum_{j,j'} \frac{_j\Theta_{j'}}{_{(j,j')}\Theta_{(j,j')}} \curlyvee_{(j,j')}
        \end{align}
        These coefficients can be compared to \eqref{eq:vdagv_dec} yielding the equations
        \begin{align}
            \sum_{k} c_{j,k} \overline{c_{j',k}} \frac{_k\Theta_{k}}{2k_{n_2}+1} = \frac{_j\Theta_{j'}}{_{(j,j')}\Theta_{(j,j')}}
            \label{eq:almost_result}
        \end{align}
        Let us further inspect the right hand side of this equation to get rid of the double evaluation of theta networks
        \begin{align}
            \frac{_j\Theta_{j'}}{_{(j,j')}\Theta_{(j,j')}} &= \delta_{j, j'} \frac{
                \prod_{a=1}^{n_1 - 1} \left( \delta_{j_a, j_{a-1} + \frac{1}{2}} \frac{2j_{a-1} + 2}{2j_{a-1} + 1} + \delta_{j_a, j_{a-1} - \frac{1}{2}} \right) (2j_{n_1}+1) 
            }{
                \prod_{a=1}^{n_1 - 1} \left( \delta_{j_a, j_{a-1} + \frac{1}{2}} \frac{2j_{a-1} + 2}{2j_{a-1} + 1} + \delta_{j_a, j_{a-1} - \frac{1}{2}} \right)^2 (2j_{n_1}+1)
            } \nonumber \\
            &= \delta_{j, j'} \frac{1}{
                \prod_{a=1}^{n_1 - 1} \left( \delta_{j_a, j_{a-1} + \frac{1}{2}} \frac{2j_{a-1} + 2}{2j_{a-1} + 1} + \delta_{j_a, j_{a-1} - \frac{1}{2}} \right)
            } \nonumber \\
            &= \delta_{j, j'} \frac{2j_{n_1}+1}{_j\Theta_j}
        \end{align}
        Finally, multiplying \eqref{eq:almost_result} with $2j_{n_1} + 1$ and using $j_{n_1} \stackrel{!}{=} k_{n_2}$ yields the result.
    \end{proof}

\begin{rem}
\label{rem:generalize}
    \rema{
        Changing the number of indices to odd valence is straightforward in the above discussion. The crucial difference comes from the fact that no invariant spin $\frac{1}{2}$ encodings with odd valence exist. Thus, the discussion needs to be generalized to higher spins $l \geq \frac{1}{2}$. \\
        To do this, Lemma \ref{prop:loopstate} needs to be generalized for higher excitations. The condition $\delta_{j,k}$ remains the same, but the spin numbers in the middle loop of \eqref{eq:loopresolve} become $l$ and $\tilde l \in \{|j-l|, |j-l|+1, ..., j+l\}$. The multiplicative factor on the right hand side of \eqref{eq:loopresolve} will depend on $\tilde{l}$. This makes the generalization of Theorem \ref{theo:master} to arbitrary spin $l$ straightforward, but tedious for high spin $l$, as every combination of $l_1, l_2 \in \{|j-l|, |j-l|+1, ..., j+l\}$ has to be considered in the basis state decompositions.\\
        The smallest non-trivial odd-valent candidate of an IPT, for instance, would be the five-valent tensor with all indices being $j=1$ representations. In the bridge state expansion, it would have the form
        \begin{align}
            \sum_{j,k} c_{j,k} \vstatefive{$j$}{$k$}{1}{1}{1}{1}{1} \qquad (j,k) \in \{(0,1), (0,2), (1, 0), (1, 1), (1,2), (2,1), (2,2)\},
        \end{align}
        with seven coefficients (two more than the six-valent spin $\frac{1}{2}$ candidate, see Appendix \ref{sec:no_6v}).
    }
\end{rem}

%% file: Chapters/Repartitions.tex
\label{sec:repartitions}
Before we study applications of \eqref{eq:master}, we want to present a number of repartitions which reduce the further analysis to only non-trivial permutations. We further sketch a systematic method of how to translate non-trivial repartitions into additional constraint equations for the coefficients $c_{(j,k)}$.

\subsection{Trivial Repartitions}
Even though the previous discussion holds for arbitrary bipartitions with index set $|A| \leq |B|$, we will show that only the balanced (or almost balanced in the case of odd $n$) bipartitions are of relevance for the study of invariant perfect tensors.
\begin{prop}
\label{prop:trvial_I}
    Let $I$ be an invariant tensor that satisfies $I^\dagger I = \lambda \mathds{1}$ for some $\lambda \in \mathds{R}$ and an implicit bipartition $A, B$. Then, the tensor $\tilde I$ which emerges from $I$ by moving one index from $A$ into $B$ is also a partial isometry $\tilde I^\dagger \tilde I = \tilde \lambda \mathds{1}$.
    \begin{proof}
        Let us write the tensor $I$ in index notation $\tensor{I}{_{B_1 ... B_b}^{A_1 ... A_a}}$. It satisfies
        \begin{align}
            \tensor{I}{^\dagger_{A_1 ... A_a}^{B_1 ... B_b}} \tensor{I}{_{B_1 ... B_b}^{A'_1 ... A'_a}} = \lambda \delta^{A'_1}_{A_1} \cdots \delta^{A'_a}_{A_a}
        \end{align}
        Assume without loss of generality that we move the $a^\text{th}$ index to $B$. The new tensor satisfies
        \begin{align}
            \tensor{\tilde I}{^\dagger_{A_1 ... A_{a-1}}^{A_a B_1 ... B_b}} \tensor{I}{_{B_1 ... B_b}^{A'_1 ... A'_{a-1}}_{A_a}} = \lambda \delta^{A'_1}_{A_1} \cdots \delta^{A'_{a-1}}_{A_{a-1}} \delta^{A_a}_{A_a} = (2j_a + 1) \lambda \delta^{A'_1}_{A_1} \cdots \delta^{A'_{a-1}}_{A_{a-1}}.
        \end{align}
        Which proves the statement with a new scaling factor $\tilde \lambda = (2j_a + 1) \lambda$ where $j_a$ denotes the spin quantum number corresponding to the $a^\text{th}$ index. 
    \end{proof}
\end{prop}
As a consequence of Lemma \ref{prop:trvial_I}, we will, without loss of generality, focus on balanced bipartitions in the following. Another set of trivial repartitions is the permutation of neighboring indices within the domain $A$ or within the image $B$.
\begin{prop}
\label{prop:trvial_II}
    Let $I$ be an invariant tensor that satisfies $I^\dagger I = \lambda \mathds{1}$ for some $\lambda \in \mathds{R}$ and an implicit bipartition $A, B$. Then the following tensors are also isometries
    \begin{enumerate}
        \item $\tilde I$ emerging from $I$ by exchanging two domain indices $A_i, A_j \in A$ 
        \item $\tilde I$ emerging from $I$ by exchanging two image indices $B_i, B_j \in B$
    \end{enumerate}
    \begin{proof}
        The proposition can be seen in the binor formulation of the isometry property. We know that $I$ satisfies
        \begin{align}
            I^\dagger I = \invtensor{$I^\dagger$}{$I$} = \lambda \identities.
        \end{align}
        \begin{enumerate}
            \item If we exchange two neighboring indices in $A$, we can calculate the isometry property
            \begin{align}
                \tilde{I}^\dagger \tilde I = \trivialone{$I^\dagger$}{$I$} = \lambda \identitiesone = \lambda \identities.
            \end{align}
            \item Analogously exchanging two neighboring indices in $B$ yields
            \begin{align}
                \tilde{I}^\dagger \tilde I = \trivialtwo{$I^\dagger$}{$I$} = \invtensor{$I^\dagger$}{$I$} = \lambda \identities.
            \end{align}
        \end{enumerate}
        For the exchange of non-neighboring indices, the above calculations can be iterated.
    \end{proof}
\end{prop}

\subsection{Non-Trivial Repartitions}
Besides the exchange of neighboring indices on the same side of the bipartition, we are left with one last repartition from nearest-neighbor permutations. This last repartition, however, has a non-trivial effect on the solution of \eqref{eq:master}, i.e. additional equations arise if we require the repartitioned tensor to also be a partial isometry. Before we study this last repartition, we point out a helpful detail arising from the fact that neighboring permutations are involutions
\begin{prop}
\label{prop:involution}
    Let $I$ be an invariant tensor that satisfies $I^\dagger I = \lambda \mathds{1}$ for some $\lambda \in \mathds{R}$ and an implicit bipartition $A, B$. Let further $\tilde I$ be a repartition of $I$ such that 
    \begin{align}
        \tilde I = \sum_{j,k} \tilde c_{j,k} \curlyvee_{j,k} = \sum_{j,k,j',k'} \tensor{P}{_{j,k}^{j',k'}} c_{j',k'} \curlyvee_{j,k}
    \end{align}
    with some involution $P$ that satisfies $P^2 = \mathds{1}$. If $\tilde I$ is a partial symmetry with $\tilde I^\dagger \tilde I = \tilde \lambda \mathds{1}$, then $\lambda = \tilde \lambda$.
    \begin{proof}
        For the sake of clarity, we introduce a combined index $\mu := (j,k)$, such that
        \begin{align}
            \tilde c_\mu \curlyvee^\mu = \tensor{P}{_\mu^\nu} c_\nu \curlyvee^\mu
        \end{align}
        Since $P$ is an involution, its eigenvalues are $\pm 1$. We can thus describe the action of $P$ on its normalized eigenvectors $\ket{p^i} = \tensor{d}{^{-1}^i_\mu} \curlyvee^\mu$
        \begin{align}
            P \left( \sum_{\mu} c_\mu ||\curlyvee^\mu|| \frac{\curlyvee^\mu}{||\curlyvee^\mu||} \right) = P \left( \sum_{\mu, i} c_\mu ||\curlyvee^\mu|| \tensor{d}{^{\mu}_i} \ket{p^i} \right) =: P (a_i\ket{p^i}) = \sum_i a_i (-1)^{x_{p_i}} \ket{p^i}
            \label{eq:involution}
        \end{align}
        where $x_{p_i} \in \{0, 1\}$ and denotes whether the $i^{th}$ eigenvalue is $+1$ or $-1$. To prove the statement, we compare the norms of $I$ and $\tilde I$. From the partial isometry property, we know that $\braket{I|I} = \Tr(I^\dagger I) = \Tr(\lambda \mathds{1}) = \lambda d$ with $d$ being the dimension of the state space on which $I^\dagger I$ is acting. Analogously, we get $\braket{\tilde I | \tilde I } = \tilde \lambda d$. We can relate the norm of $\tilde I$ to $\lambda$ using \eqref{eq:involution}
        \begin{align}
            \braket{ I | I} &= \sum_\mu |c_\mu|^2 ||\curlyvee^\mu||^2 = \sum_i |a_i|^2 \\
            \braket{\tilde I | \tilde I} &= \Tr\left( P \left( \sum_{\mu} c_\mu ||\curlyvee^\mu|| \frac{\curlyvee^\mu}{||\curlyvee^\mu||} \right)^\dagger P \left( \sum_{\mu} c_\mu ||\curlyvee^\mu|| \frac{\curlyvee^\mu}{||\curlyvee^\mu||} \right) \right) = \sum_{i} |a_i (-1)^{x_{p_i}}|^2 = \sum_{i} |a_i|^2 
        \end{align}
    \end{proof}
    This implies $\tilde \lambda = \lambda$.
\end{prop}
\begin{prop}
\label{prop:non_triv_repart}
    Let $I$ be an invariant tensor that satisfies $I^\dagger I = \lambda \mathds{1}$ for some $\lambda \in \mathds{R}$ and an implicit bipartition $A, B$. We define the tensor $\tilde I$ which arises from changing the last indices in $A$ and $B$, respectively, i.e.
    \begin{align}
        \tilde I = \bipartstate{$I$}
    \end{align}
     $\tilde I$ is also a partial isometry, if and only if 
    \begin{align}
        \bipartproduct{$I^\dagger$}{$I$} - \bipartproductII{$I^\dagger$}{$I$} - \bipartproductIII{$I^\dagger$}{$I$} = 0
        \label{eq:non_triv_repart}
    \end{align}
    \begin{proof}
        The statement results from the isometry condition 
        \begin{align}
            &\tilde I^\dagger \tilde I \stackrel{!}{=} \tilde \lambda \mathds{1} \nonumber \\
            \iff & \bipartproductIV{$I^\dagger$}{$I$} = \tilde \lambda \identities \nonumber \\
            \iff &\invtensor{$I^\dagger$}{$I$} + \bipartproduct{$I^\dagger$}{$I$} - \bipartproductII{$I^\dagger$}{$I$} - \bipartproductIII{$I^\dagger$}{$I$} = \tilde \lambda \identities
            \label{eq:non_triv_repart_proof}
        \end{align}
        with $\tilde \lambda \in \mathds{R}$ and where we used the binor identity to get rid of the crossings. From Lemma \ref{prop:involution} we know that $\lambda = \tilde \lambda$. As we know that $I$ is a partial isometry, we can cancel the first term of the left hand side with the right hand side of \eqref{eq:non_triv_repart_proof} and deduce the statement.
    \end{proof}
\end{prop}
Note that we necessarily pull indices with the Levi-Civita tensor $\epsilon$ in order to use the binor identity in \eqref{eq:non_triv_repart_proof}. It can be straight-forwardly shown that the isometry properties of $I$ and $\epsilon I \epsilon^\dagger$ are equivalent which allows us to do so. 

In general, the condition \eqref{eq:non_triv_repart} turns out to be difficult to further break down in general. However, it can always be evaluated by decomposing the remaining terms again into bridge states. The equation \eqref{eq:non_triv_repart} might put new constraints on the previously found coefficients. This procedure needs to be iterated over all combinations of trivial permutations $P_{i,j}$ preceding the critical permutation $P_{n, n+1}$. Since they do not commute, in general, new constraints might arise for the coefficients $c_{j,k}$. We illustrate this method on two examples in appendices \ref{sec:no_4v} and \ref{sec:no_6v}, respectively.

%% file: Chapters/no_4v.tex
\label{sec:no_4v}
    \begin{proof}[Second Proof of Theorem \ref{theo:no_4v_I}]
        We proceed by contradiction. Let us assume there would exist a 4-valent invariant perfect tensor $I$ with spin $\frac{1}{2}$. The bridge decomposition of $I$ reads
        \begin{align}
            I = \sum_{j,k} c_{j,k} \curlyvee_{(j,k)}
        \end{align}
        For spin $\frac{1}{2}$ and a fixed bipartition, there are only two basis states with $j = k = \left( \frac{1}{2}, 0 \right)$ and $j = k = \left( \frac{1}{2}, 1 \right)$, which we will denote $\curlyvee_0$ and $\curlyvee_1$. By Theorem \ref{theo:master}, the coefficients have to satisfy the set of equations
        \begin{align}
            |c_1|^2 = \lambda \qquad &\iff \qquad c_1 = \sqrt{\lambda} e^{i\phi_1} \\
            |c_0|^2 = \frac{\lambda}{4} \qquad &\iff \qquad c_0 = \frac{\sqrt{\lambda}}{2} e^{i\phi_0}.
        \end{align}
        If we do a repartition changing the neighboring indices across the bridge, we map $I$ to $\tilde I$ which we can write as
        \begin{align}
            \tilde I &= \sum_{j \in \{0,1\}} \tilde c_j \curlyvee_j := \sum_{j \in \{0,1\}} c_j \fourvrepart{$j$} = \sum_{j \in \{0,1\}} c_j \left( \curlyvee_j - \fourvloop{$j$} \right) \nonumber \\
            &= c_0 \left( \curlyvee_0 - \mathds{1} \right) + c_1 \left( \curlyvee_1 - \frac{3}{2} \mathds{1} \right) \label{eq:first_repart}
        \end{align}
        We can insert the bridge state decomposition of $\mathds{1} = \frac{1}{2} \curlyvee_0 + \curlyvee_1$ (cf. Lemma \ref{prop:1_dec}) to read off the new coefficients $\tilde c_j$
        \begin{align}
            \tilde c_0 = \frac{1}{2} c_0 - \frac{3}{4} c_1 \qquad \qquad \tilde c_1 = -c_0 - \frac{1}{2} c_1 \label{eq:non_triv_trafo}
        \end{align}
        where we used that $_{(\frac{1}{2}, 1)}\Theta_{(\frac{1}{2}, 1)} = _{(\frac{1}{2}, 1, \frac{1}{2})}\Theta_{(\frac{1}{2}, 1, \frac{1}{2})} = 3$, $_{(\frac{1}{2}, 0)}\Theta_{(\frac{1}{2}, 0)} = 2$ and $_{(\frac{1}{2}, 0, \frac{1}{2})}\Theta_{(\frac{1}{2}, 0, \frac{1}{2})} = 4$. The twiddled coefficients $\tilde c_j$ have to satisfy the same equations as before 
        \begin{align}
            |\tilde c_1|^2 &= \tilde \lambda = |c_0|^2 + \frac{1}{4} |c_1|^2 + \Real(\overline{c_0} c_1) = \frac{\lambda}{2} \left( 1 + \cos(\phi_1 - \phi_0) \right) \nonumber \\
            &\iff \frac{\tilde \lambda}{\lambda} = \frac{1}{2} + \frac{1}{2}\cos(\phi_1 - \phi_0) \label{eq:v4_condI} \\
            |\tilde c_0|^2 &= \frac{\tilde \lambda}{4} = \frac{1}{4} |c_0|^2 + \frac{9}{16} |c_1|^2 - \frac{3}{4} \Real(\overline{c_0} c_1) = \frac{\lambda}{8} \left( 5 - 3\cos(\phi_1 - \phi_0) \right) \nonumber \\
            &\iff \frac{\tilde \lambda}{\lambda} = \frac{5}{2} - \frac{3}{2} \cos(\phi_1 - \phi_0),
            \label{eq:v4_condII}
        \end{align}
        With the equations \eqref{eq:v4_condI} and \eqref{eq:v4_condII}, we can solve for $\frac{\tilde \lambda}{\lambda}$ as well as $\Delta \phi := \phi_1 - \phi_0$
        \begin{align}
            1 + \cos(\Delta\phi) = 5 - 3 \cos(\Delta\phi) \qquad \iff \qquad \Delta\phi = 0 \mod 2\pi
            \label{eq:resultI}
        \end{align}
        which implies $\frac{\tilde \lambda}{\lambda} = 1$ and is consistent with Lemma \ref{prop:involution}. Lastly, we consider a repartition that swaps the two indices on the right of bridge state of $I$, then performs the repartition as in \eqref{eq:first_repart} and finally swaps the two right indices again. Since the neighboring permutations do not commute, we get a non-trivial new repartition $\tilde{ \tilde I}$ via
        \begin{align}
            \tilde{ \tilde I} &= \sum_{j \in \{0,1\}} \tilde{\tilde{c_j}} \curlyvee_j = \sum_{j \in \{0,1\}} c_j \fourvrepartII{$j$}. 
            \label{eq:second_repart}
        \end{align}
        We could again use the binor identity to determine the coefficients $\tilde{\tilde{c_j}}$. Instead, we just use the knowledge from \eqref{eq:non_triv_trafo} and the transformation
        \begin{align}
            c_{new} = 
            \begin{pmatrix} 
                -1 & 0 \\
                0 & 1
            \end{pmatrix}
            c_{old}
        \end{align}
        coming from the trivial repartition which changes the two indices on the right. This can be straight-forwardly shown by using the binor identity. Finally, the new coefficients can be calulated by composing the basis transformations accordingly
        \begin{align}
            \tilde{\tilde{c_j}} = 
            \tensor{ \left[ \begin{pmatrix} 
                -1 & 0 \\
                0 & 1
            \end{pmatrix}
            \begin{pmatrix} 
                \frac{1}{2} & -\frac{3}{4} \\
                -1 & -\frac{1}{2}
            \end{pmatrix}
            \begin{pmatrix} 
                -1 & 0 \\
                0 & 1
            \end{pmatrix} \right]}{_j ^k}
            c_k \qquad \implies \tilde{\tilde{c_0}} = \frac{1}{2} c_0 + \frac{3}{4} c_1 \qquad \tilde{\tilde{c_1}} = c_0 - \frac{1}{2} c_1
        \end{align}
        Following the same steps as in (\ref{eq:v4_condI} - \ref{eq:resultI}), we can again fix the relative phase between $c_0$ and $c_1$
        \begin{align}
            1 - \cos(\Delta\phi) = 5 + 3 \cos(\Delta\phi) \qquad \iff \qquad \Delta\phi = \pi \mod 2\pi
        \end{align}
        which is a contradiction with \eqref{eq:resultI}.
    \end{proof}

%% file: Chapters/no_6v.tex
\label{sec:no_6v}
\begin{proof}[Proof of Theorem \ref{theo:no_6v}]
        We proceed by contradiction. Let us assume there would exist a 6-valent invariant perfect tensor $I$ with spin $\frac{1}{2}$. The bridge decomposition of $I$ reads
        \begin{align}
            I = \sum_{j,k} c_{j,k} \curlyvee_{(j,k)},
        \end{align}
        where we reduce the discussion without loss of generality to balanced bridge states (cf. \ref{prop:trvial_I}). The basis states are then drawn from the index set 
        \begin{align}
            (j,k) = (j_2, j_3=k_3, k_2) \in \Bigg\{\left(1, \frac{3}{2}, 1\right), \left(1, \frac{1}{2}, 1\right), \left(0, \frac{1}{2}, 1\right), \left(1, \frac{1}{2}, 0\right), \left(0, \frac{1}{2}, 0\right) \Bigg\}.
            \label{eq:v6_indices}
        \end{align}
        Note that we omitted again the first spin which is always $j_1 = k_1 = \frac{1}{2}$. We can again employ the master equation \eqref{eq:master} for every possible combination $j, j'$ in \eqref{eq:v6_indices}. Starting with $j = j' = j_{max}$, we can recover the result of Corollary \ref{cor:max_coeff}
        \begin{align}
            |c_{j_{max}, j_{max}}|^2 = \lambda \qquad \implies c_{j_{max}, j_{max}} = \sqrt{\lambda} e^{i \phi} 
            \label{eq:v6_eq1}
        \end{align}
        The other two symmetric equations $j=j' = \left( \frac{1}{2}, 1, \frac{1}{2} \right)$ and $j = j' = \left( \frac{1}{2}, 0, \frac{1}{2} \right)$ read
        \begin{align}
            \frac{9}{4} |c_{(1,\frac{1}{2}, 1)}|^2 + 3|c_{(1,\frac{1}{2}, 0)}|^2 &= \lambda \\
            4 |c_{(0,\frac{1}{2}, 0)}|^2 + 3 |c_{(0,\frac{1}{2}, 1)}|^2 &= \lambda
        \end{align}
        using ${}_{(\frac{1}{2}, 1, \frac{1}{2})}\Theta_{(\frac{1}{2}, 1, \frac{1}{2})} = 3$ and ${}_{(\frac{1}{2}, 0, \frac{1}{2})}\Theta_{(\frac{1}{2}, 0, \frac{1}{2})} = 4$. The last non-trivial equation comes from \eqref{eq:master} by plugging in $j = \left( \frac{1}{2}, 0, \frac{1}{2} \right)$ and $j' = \left( \frac{1}{2}, 1, \frac{1}{2} \right)$
        \begin{align}
            4 \, c_{(0,\frac{1}{2}, 0)} \overline{ c_{(1,\frac{1}{2}, 0)} } + 3 \, c_{(0,\frac{1}{2}, 1)} \overline{ c_{(1,\frac{1}{2}, 1)} } = 0
            \label{eq:v6_eq4}
        \end{align}
        As we only have 4 equations, but 5 coefficients to be determined, we can leave one of the coefficients as a free parameter. We will choose $c_{(1,\frac{1}{2},1)} = A e^{i \rho}$, with $A \in \mathds{R}^+, \rho \in [0, 2 \pi]$. Solving equations (\ref{eq:v6_eq1} - \ref{eq:v6_eq4}), we get
        \begin{align}
            c_{(1,\frac{3}{2}, 1)} &= \sqrt{\lambda} e^{i \phi} \label{eq:parameterI} \\
            c_{(1,\frac{1}{2}, 1)} &= A e^{i \rho} \label{eq:parameterII} \\
            c_{(1,\frac{1}{2}, 0)} &= \sqrt{\frac{\lambda}{3} - \frac{3}{4} A^2 } \, e^{i \chi} \label{eq:parameterIII} \\
            c_{(0,\frac{1}{2}, 0)} &= \frac{3}{4} A e^{i\psi} \label{eq:parameterIV} \\
            c_{(0,\frac{1}{2}, 1)} &= \sqrt{\frac{\lambda}{3} - \frac{3}{4} A^2} e^{i(\rho + \psi - \chi + \pi)}
            \label{eq:parameterV}
        \end{align}
        Before considering any repartition, we have found a subspace of invariant tensors $\mathcal{I}(A, \phi, \rho, \psi, \chi)$ as suitable candidates for perfect tensors parameterized by one amplitude $A$ and four phases $\phi, \rho, \chi, \psi$. If there was a 6-valent invariant perfect tensor, we would find a subspace $\varnothing \neq \mathcal{I}' \subset \mathcal{I}$ that would be stable with respect to all possible repartitions. We will show that necessarily $\mathcal{I}' = \varnothing$. Let us start with the repartition $P^*$ exchanging the indices across the bridge. The new coefficients $\tilde c_{j,k}$ defined by the decomposition into bridge states
        \begin{align}
            \sixvrepart{$j$}{$k$} = \sum_{j,k} \tilde c_{j,k} \curlyvee_{j,k},
            \label{eq:v6_repart}
        \end{align}
        can be found analogously to the second proof of Theorem \ref{theo:no_4v_I}. The basis transformation reads
        \begin{align}
            \begin{pmatrix}
                \tilde c_{(1,\frac{3}{2},1)} \\
                \tilde c_{(1,\frac{1}{2},1)} \\
                \tilde c_{(1,\frac{1}{2},0)} \\
                \tilde c_{(0,\frac{1}{2},1)} \\
                \tilde c_{(0,\frac{1}{2},0)}
            \end{pmatrix}
            =
            \begin{pmatrix}
                -\frac{1}{3} & -1 & 0 & 0 & 0 \\
                -\frac{8}{9} & \frac{1}{3} & 0 & 0 & 0 \\
                0 & 0 & 1 & 0 & 0 \\
                0 & 0 & 0 & 1 & 0 \\
                0 & 0 & 0 & 0 & -1
            \end{pmatrix}
            \begin{pmatrix}
                c_{(1,\frac{3}{2},1)} \\
                c_{(1,\frac{1}{2},1)} \\
                c_{(1,\frac{1}{2},0)} \\
                c_{(0,\frac{1}{2},1)} \\
                c_{(0,\frac{1}{2},0)}
            \end{pmatrix}
        \end{align}
        Since the master equation does not see which partition we are currently looking at, the same parameterization (\ref{eq:parameterI} - \ref{eq:parameterV}) holds for the twiddled coefficients $\tilde c_{j,k}$. We will now go through every one of them and apply the additional conditions. Starting with (\ref{eq:parameterIV}, \ref{eq:parameterIII}, \ref{eq:parameterV}), we get
        \begin{align}
            \tilde c_{(0, \frac{1}{2}, 0)} = \frac{3}{4} \tilde A e^{i \tilde \psi} = - c_{(0, \frac{1}{2}, 0)} = -\frac{3}{4} A e^{i \psi} \qquad &\iff \tilde A = A, \quad \tilde \psi = \psi + \pi \\
            \tilde c_{(1, \frac{1}{2}, 0)} = \sqrt{\frac{\lambda}{3} - \frac{3}{4} A^2 } \, e^{i \chi} = c_{(1, \frac{1}{2}, 0)} = \sqrt{\frac{\tilde \lambda}{3} - \frac{3}{4} \tilde A^2 } \, e^{i \tilde \chi} \qquad &\iff \tilde \lambda = \lambda, \quad \tilde \chi = \chi \\
            \tilde c_{(0, \frac{1}{2}, 1)} = \sqrt{\frac{\lambda}{3} - \frac{3}{4} A^2} e^{i(\rho + \psi - \chi + \pi)} = c_{(0, \frac{1}{2}, 1)} = \sqrt{\frac{\tilde \lambda}{3} - \frac{3}{4} \tilde A^2} e^{i(\tilde \rho + \tilde \psi - \tilde \chi + \pi)} \qquad &\iff \tilde \rho = \rho + \pi
        \end{align}
        This fixes all new parameters except $\tilde \phi$. Next, we use \eqref{eq:parameterI}
        \begin{align}
            |\tilde c_{(1, \frac{3}{2}, 1)}|^2 &= \left( - \frac{1}{3} \overline{c_{(1, \frac{3}{2}, 1)}} - \overline{c_{(1, \frac{1}{2}, 1)}} \right)\left( - \frac{1}{3} c_{(1, \frac{3}{2}, 1)} - c_{(1, \frac{1}{2}, 1)} \right) = \frac{1}{9} \lambda + A^2 + \frac{2}{3} A \sqrt{\lambda} \cos(\phi - \rho) = \lambda \nonumber \\
            \iff 0 &= 9 A^2 + 6 A \sqrt{\lambda} \cos(\phi - \rho) -8 \lambda \nonumber \\
            \iff A &= \frac{\sqrt{\lambda}}{3} \left( \sqrt{ \cos^2(\phi - \rho) + 8 } - \cos(\phi - \rho) \right),
            \label{eq:A_amplitude}
        \end{align}
        which fixes $A$. We can also fix $\tilde \phi$ by comparing the phases
        \begin{align}
            \sqrt{\lambda} \cos\left( \tilde \phi \right) = -\frac{\sqrt{\lambda}}{3} \cos(\phi) - A \cos(\rho)
            \label{eq:phi_phase}
        \end{align}
        Lastly, we use \eqref{eq:parameterII} and the results from before
        \begin{align}
            \tilde c_{(1, \frac{1}{2}, 1)} &= A e^{i(\rho+\pi)} = -\frac{8}{9} c_{(1, \frac{3}{2}, 1)} + \frac{1}{3} c_{(1, \frac{1}{2}, 1)} = -\frac{8}{9} \sqrt{\lambda} e^{i \phi} + \frac{1}{3} A e^{i \rho} \nonumber \\
            \iff A e^{i \rho} &= \frac{2}{3} \sqrt{\lambda} e^{i \phi},
        \end{align}
        which yields $A = \frac{2}{3} \sqrt{\lambda}$ and $\rho = \phi$, which is consistent with \eqref{eq:A_amplitude}. Also, we find $\tilde \phi = \phi + \pi$ from \eqref{eq:phi_phase}. Summing up, the repartition \eqref{eq:v6_repart} relates all the new degrees of freedoms with respect to the old ones and additionally fixes the formerly free amplitude $A$ and one of the phases $\rho = \phi$. We get the updated coefficients
        \begin{align}
            c_{(1,\frac{3}{2}, 1)} &= \sqrt{\lambda} e^{i \phi} \qquad
            c_{(1,\frac{1}{2}, 1)} = \frac{2}{3} \sqrt{\lambda} e^{i \phi} \qquad c_{(0,\frac{1}{2}, 0)} = \frac{\sqrt{\lambda}}{2} e^{i\psi} \\
            c_{(1,\frac{1}{2}, 0)} &= 0 \qquad
            c_{(0,\frac{1}{2}, 1)} = 0
        \end{align}
        The last repartition, we consider is inspired by the critical repartition \eqref{eq:second_repart} from the proof for 4-valent tensors, which first exchanges the two bottom indices on the right side, then the two indices across the bridge and finally exchange the same two indices on the right hand side again. The transformation matrix for the coefficients reads
        \begin{align}
            \begin{pmatrix}
                \tilde{\tilde{c}}_{(1,\frac{3}{2},1)} \\
                \tilde{\tilde{c}}_{(1,\frac{1}{2},1)} \\
                \tilde{\tilde{c}}_{(1,\frac{1}{2},0)} \\
                \tilde{\tilde{c}}_{(0,\frac{1}{2},1)} \\
                \tilde{\tilde{c}}_{(0,\frac{1}{2},0)}
            \end{pmatrix}
            &=
            \begin{pmatrix}
                1 & 0 & 0 & 0 & 0 \\
                0 & -\frac{1}{2} & -1 & 0 & 0 \\
                0 & -\frac{3}{4} & \frac{1}{2} & 0 & 0 \\
                0 & 0 & 0 & -\frac{1}{2} & -1 \\
                0 & 0 & 0 & -\frac{3}{4} & \frac{1}{2}
            \end{pmatrix}
            \begin{pmatrix}
                -\frac{1}{3} & -1 & 0 & 0 & 0 \\
                -\frac{8}{9} & \frac{1}{3} & 0 & 0 & 0 \\
                0 & 0 & 1 & 0 & 0 \\
                0 & 0 & 0 & 1 & 0 \\
                0 & 0 & 0 & 0 & -1
            \end{pmatrix}
            \begin{pmatrix}
                1 & 0 & 0 & 0 & 0 \\
                0 & -\frac{1}{2} & -1 & 0 & 0 \\
                0 & -\frac{3}{4} & \frac{1}{2} & 0 & 0 \\
                0 & 0 & 0 & -\frac{1}{2} & -1 \\
                0 & 0 & 0 & -\frac{3}{4} & \frac{1}{2}
            \end{pmatrix}
            \begin{pmatrix}
                c_{(1,\frac{3}{2},1)} \\
                c_{(1,\frac{1}{2},1)} \\
                c_{(1,\frac{1}{2},0)} \\
                c_{(0,\frac{1}{2},1)} \\
                c_{(0,\frac{1}{2},0)}
            \end{pmatrix} \nonumber \\
            &= 
            \begin{pmatrix}
                -\frac{1}{3} & \frac{1}{2} & 1 & 0 & 0 \\
                \frac{4}{9} & -\frac{1}{6} & -1 & 0 & 0 \\
                \frac{2}{3} & -\frac{1}{4} & \frac{1}{2} & 0 & 0 \\
                0 & 0 & 0 & -\frac{1}{2} & 1 \\
                0 & 0 & 0 & -\frac{3}{4} & -\frac{1}{2}
            \end{pmatrix}
            \begin{pmatrix}
                c_{(1,\frac{3}{2},1)} \\
                c_{(1,\frac{1}{2},1)} \\
                c_{(1,\frac{1}{2},0)} \\
                c_{(0,\frac{1}{2},1)} \\
                c_{(0,\frac{1}{2},0)}
            \end{pmatrix},
        \end{align}
        which generates a contradiction, for instance if we regard $0=\tilde{\tilde{c}}_{(0, \frac{1}{2}, 1)} = c_{(0, \frac{1}{2}, 0)} = \frac{\sqrt{\lambda}}{2} e^{i\psi}$. This would imply $\tilde{\tilde{\lambda}} = 0$, which is a contradiction to the perfectness assumption.
    \end{proof}